\documentclass[amssymb,reprint,pra,superscriptaddress]{revtex4-1}
\usepackage[dvips]{graphicx} % for figures
\usepackage{amsfonts}
\usepackage{amssymb}
\usepackage{amscd}
\usepackage{amsmath}    % need for subequations
\usepackage{enumerate}
\usepackage{epsfig}
\usepackage{subfigure}
\usepackage{xcolor}
\usepackage{amsthm}
\usepackage{framed}
\usepackage{multirow}
\usepackage{epstopdf}

\newtheorem{theorem}{Theorem}
\newtheorem{lemma}{Lemma}
\newtheorem{corollary}{Corollary}

\newtheorem{proposition}{Proposition}
\newtheorem{definition}{Definition}

\newcommand{\bra}[1]{\mbox{$\left\langle #1 \right|$}}
\newcommand{\ket}[1]{\mbox{$\left| #1 \right\rangle$}}

\newcommand{\comments}[1]{}

\newcommand{\CZ}{\textrm{CZ}}

\begin{document}
\preprint{APS/123-QED}
\title{Detecting multipartite entanglement structure with minimal resources}
%\title{Efficient detection of multipartite entanglement structure}
\date{\today}% It is always \today, today,
\author{You Zhou}
%\email{zhouyou14@mails.tsinghua.edu.cn}
\affiliation{Center for Quantum Information, Institute for Interdisciplinary Information Sciences, Tsinghua University, Beijing 100084, China}

\author{Qi Zhao}
%\email{zhaoq14@mails.tsinghua.edu.cn}
\affiliation{Center for Quantum Information, Institute for Interdisciplinary Information Sciences, Tsinghua University, Beijing 100084, China}

\author{Xiao Yuan}
%\email{xiao.yuan.ph@gmail.com}
\affiliation{Department of Materials, University of Oxford, Parks Road, Oxford OX1 3PH, United Kingdom}

\author{Xiongfeng Ma}
\email{xma@tsinghua.edu.cn}
\affiliation{Center for Quantum Information, Institute for Interdisciplinary Information Sciences, Tsinghua University, Beijing 100084, China}
%\affiliation{Corresponding Author, xma@tsinghua.edu.cn}

\begin{abstract}
Recently, there are tremendous developments on the number of controllable qubits in several quantum computing systems. For these implementations, it is crucial to determine the entanglement structure of the prepared multipartite quantum state as a basis for further information processing tasks. In reality, evaluation of a multipartite state is in general a very challenging task owing to the exponential increase of the Hilbert space with respect to the number of system components. In this work, we propose a systematic method using very few local measurements to detect multipartite entanglement structures based on the graph state --- one of the most important classes of quantum states for quantum information processing. Thanks to the close connection between the Schmidt coefficient and quantum entropy in graph states, we develop a family of efficient witness operators to detect the entanglement between subsystems under any partitions and hence the entanglement intactness. We show that the number of local measurements equals to the chromatic number of the underlying graph, which is a constant number, independent of the number of qubits. In reality, the optimization problem involved in the witnesses can be challenging with large system size. For several widely-used graph states, such as 1-D and 2-D cluster states and the Greenberger-Horne-Zeilinger state, by taking advantage of the area law of entanglement entropy, we derive analytical solutions for the witnesses, which only employ two local measurements. Our method offers a standard tool for entanglement structure detection to benchmark multipartite quantum systems.
\end{abstract}
%\pacs{}% PACS, the Physics and Astronomy
%\keywords{}%Use showkeys class option if keyword
\maketitle

\section{Introduction}
Entanglement is an essential resource for many quantum information tasks \cite{Horodecki2009entanglement}, such as quantum teleportation \cite{Bennett1993Teleporting}, quantum cryptography \cite{Bennett84cryptography,Ekert1991cryptography}, non-locality test \cite{Brunner2014nonlocality}, quantum computing \cite{Nielsen2011Quantum}, quantum simulation \cite{Lloyd1996Simulators} and quantum metrology \cite{Wineland1992squeezing,Giovannetti2006Metrology}. Tremendous efforts have been devoted to the realization of multipartite entanglement in various systems \cite{Monz2011Entanglement,Britton2012trapped,Friis2018Observation,Song2017Entanglement,Gong2019Genuine,Wang2016Entanglement,chen2017observation,12photon, Toth2014Detecting,Luo2017Deterministic,Lange2018atomic}, which provide the foundation for small- and medium-scale quantum information processing in near future and will eventually pave the way to universal quantum computing. In order to build up a quantum computing device, it is crucial to first witness multipartite entanglement. So far, genuine multipartite entanglement has been demonstrated and witnessed in experiment with a small amount of qubits in different realizations, such as 14-ion-trap-qubit \cite{Monz2011Entanglement}, 12-superconducting-qubit \cite{Gong2019Genuine}, and 12-photon-qubit systems \cite{12photon}.

%In practical quantum hardware, unavoidable coupling to the environment undermines the fidelity between the prepared state and the target one. For example, the state-of-the-art 10-superconducting-qubit \cite{Gong2019Genuine} and the 12-photon \cite{12photon} preparations only achieve the fidelity of $55.4\%$ and $57.2\%$, respectively, which just exceed the threshold 50\% for the certification of genuine entanglement.
In practical quantum hardware, the unavoidable coupling to the environment undermines the fidelity between the prepared state and the target one. Taking the Greenberger-Horne-Zeilinger (GHZ) state for example, the state-of-the-art 10-superconducting-qubit \cite{Song2017Entanglement} and the 12-photon \cite{12photon} preparations only achieve the fidelity of $66.8\%$ and $57.2\%$, respectively, which just exceed the threshold 50\% for the certification of genuine entanglement.
As the system size becomes larger, see for instance, Google's a $72$-qubit chip \footnote{https://www.sciencenews.org/article/google-moves-toward-quantum-supremacy-72-qubit-computer} and IonQ's a $79$-qubit system \footnote{https://physicsworld.com/a/ion-based-commercial-quantum-computer-is-a-first/}, it is an experimental challenge to create genuine multipartite entanglement. Nonetheless, even without global genuine entanglement as the target state possesses, the experimental prepared state might still have fewer-body entanglement within a subsystem and/or among distinct subsystems \cite{Dur2000Three,Acin2001Classification,Guhne2005Multipartite}. The study of lower-order entanglement, which can be characterized by the detailed entanglement structures \cite{Huber2013Structure,Shahandeh2014Structural,Lu2018Structure}, is important for quantum hardware development, because it might reveal the information on unwanted couplings to the environment and acts as a benchmark of the underlying system. Moreover, the certified lower-order entanglement among several subsystems could be still useful for some quantum information tasks.

Considering an $N$-partite quantum system and its partition into $m$ subsystems ($m\le N$), the entanglement structure indicates how the subsystems are entangled with each other. Each subsystem corresponds to a subset of the whole quantum system. For instance, we can choose each subsystem to be each party (i.e., $m=N$), and then the entanglement structure indicates the entanglement between the $N$ parties. In some specific systems, such as distributed quantum computing \cite{Cirac1999Distributed}, quantum networks \cite{Kimble2008internet} or atoms in a lattice, the geometric configuration can naturally determine the system partition (see FIG.~\ref{Fig:distQC} for an illustration).
In other cases, one might not need to specify the partition in the beginning. By going through all possible partitions, one can investigate higher level entanglement structures, such as entanglement intactness (non-separability)  \cite{Guhne2005Multipartite,Lu2018Structure}, which quantifies how many pieces in the $N$-partite state are separated.

\begin{figure}[tbhp!]
\centering \includegraphics[width=0.4\textwidth]{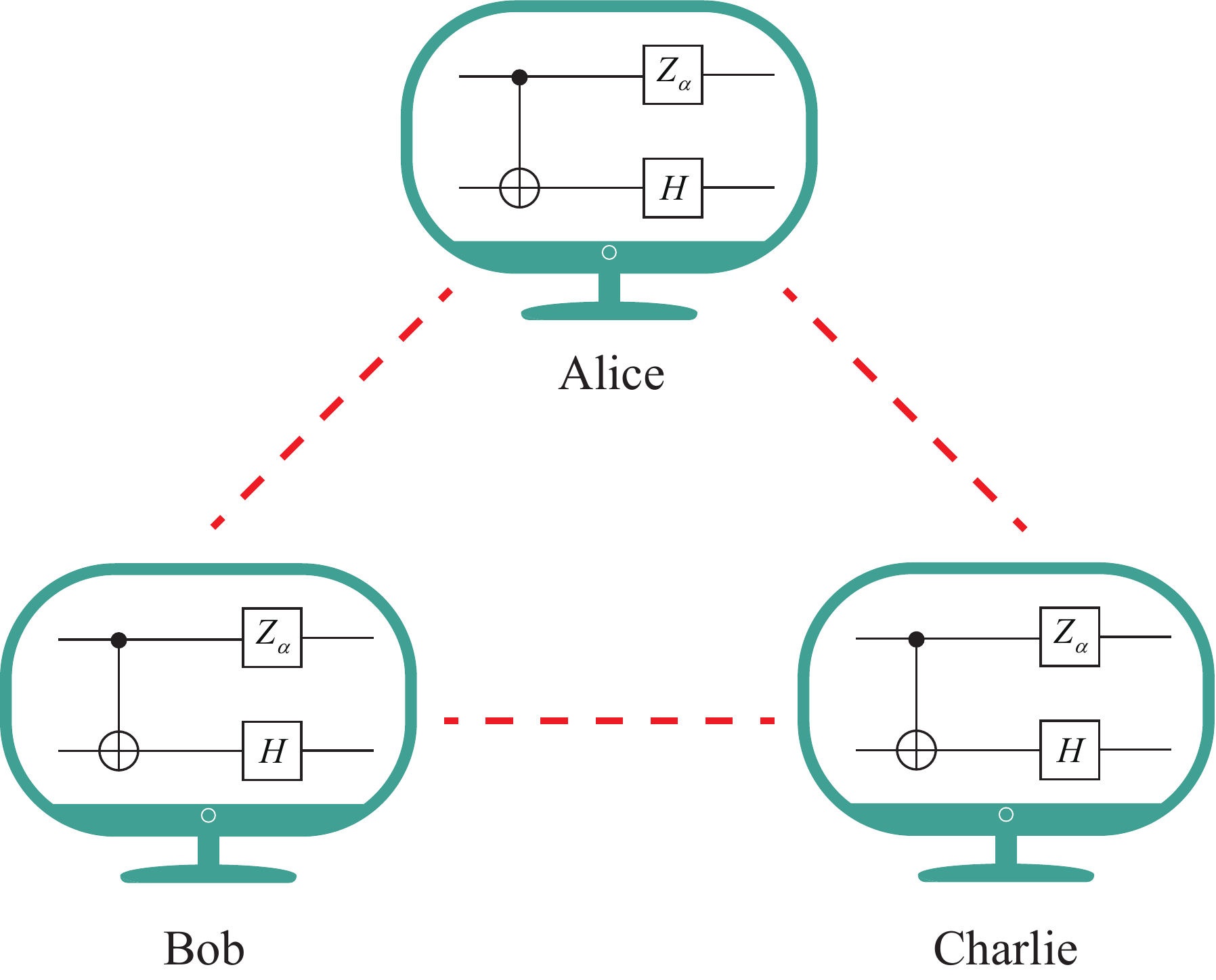}
\caption{A distributed quantum computing scenario. Three remote (small) quantum processors, owned by Alice, Bob and Charlie, are connected by quantum links. Each of them possesses a few of qubits and performs quantum operations. In this case, the partition of the whole quantum system is determined by the locations of these processors. In order to perform global quantum operations involving multiple processors, entanglement among the processors are generally required. Thus, it is essential to benchmark the entanglement structure on this network.
}\label{Fig:distQC}
\end{figure}

Multipartite entanglement structure detection is generally a challenging task. Naively, one can perform state tomography on the system. As the system size increases, tomography becomes infeasible due to the exponential increase of the Hilbert space. Entanglement witness \cite{TERHAL2001witness,GUHNE2009detection,Friis2019Reviews}, on the other hand, provides an elegant solution to multipartite entanglement detection. In literature, various witness operators have been proposed to detect different types of quantum states, generally requiring a polynomial number of measurements with respect to the system size \cite{Guhne2007Toolbox,You2019Sym}. Interestingly, a constant number of local measurement settings are shown to be sufficient to detect genuine entanglement for stabilizer states \cite{Toth2005Detecting,Knips2016Multipartite}.
%One might wonder whether such witness technique can be applied to entanglement structure detection.
Compared with genuine entanglement, multipartite entanglement structure still lacks a systematic exploration,  due to the rich and complex structures of $N$-partite system.
Recently, positive results have been achieved for detecting entanglement structures of GHZ-like states with two measurement settings \cite{Lu2018Structure} and the entanglement of a specific 1-D cluster state of the 16-qubit superconducting quantum processor $ibmqx5$ machine from the IBM cloud \cite{ibm16entanglement}. %Detecting the entanglement structure becomes an essential step for benchmarking quantum processors like $ibmqx5$.
Unfortunately, it remains an open problem of efficient entanglement structure detection of general multipartite quantum states.

In this work, we propose a systematic method to witness the entanglement structure based on graph states. Note that the graph state \cite{Briegel2001Persistent,Hein2006Graph} is one of the most important classes of multipartite states for quantum information processing, such as measurement-based quantum computing \cite{onewayQC,Raussendorf2003Measurement}, quantum routing and quantum networks \cite{Kimble2008internet}, quantum error correction \cite{Werner2001error}, and Bell nonlocality test \cite{Guhne2005Bell}. It is also related to the symmetry-protected topological order in condensed matter physics \cite{Bei2019Meets}.
%Note that the graph state \cite{Briegel2001Persistent,Hein2006Graph} is one of the most important classes of multipartite states for quantum information processing \cite{onewayQC,Raussendorf2003Measurement}.
%Note that the graph state \cite{Briegel2001Persistent,Hein2006Graph} is one of the most important classes of multipartite states for quantum information processing, such as measurement-based quantum computing \cite{onewayQC,Raussendorf2003Measurement}, quantum routing and quantum networks \cite{Perseguers2013Distribution}, quantum error correction \cite{Werner2001error}, and Bell nonlocality test \cite{Scarani2005Nonlocality,Guhne2005Bell}. It is also related to the symmetry-protected topological order in condensed matter physics \cite{Bei2019Meets}.
Typical graph states include cluster states, GHZ state, and the states involved in the encoding process of the 5-qubit Steane code and the concatenated $[7,1,3]$-CSS-code \cite{Hein2006Graph}.

The main idea of our entanglement structure detection method runs as follows. First, with the close connection between the maximal Schmidt coefficient and quantum entropy, we upper-bound the fidelity of fully- and bi-separable states. These bounds are directly related to the entanglement entropy of the underlying graph state with respect to certain bipartition. Then, inspired by the genuine entanglement detection method \cite{Toth2005Detecting}, we lower-bound the fidelity between the unknown prepared state and the target graph state, with local measurements corresponding to the stabilizer operators of the graph state. Finally, by comparing theses fidelity bounds, we can witness the entanglement structures, such as the (genuine multipartite) entanglement between any subsystem partitions, and hence the entanglement intactness.

Our detection method for entanglement structures based on graph states is presented in Theorem \ref{Th:main} and \ref{Th:msep}, which only involves $k$ local measurements. Here, $k$ is the chromatic number of the corresponding graph, typically, a small constant independent of the number of qubits.
%In order to obtain the corresponding bounds related to the proposed witnesses, one needs to solve optimization problems, which grow exponentially with the system size. In general, the optimization is a challenge when the system size is large.
For several common graph states, 1-D and 2-D cluster states and the GHZ state, we construct witnesses with only $k=2$ local measurement settings, and derive analytical solutions to the optimization problem. These results are shown in Corollaries \ref{Th:GHZ} to \ref{Th:2Dmsep}. The proofs of propositions and theorems are left in Methods, and the proofs of Corollaries \ref{Th:GHZ} to \ref{Th:2Dmsep} are presented in Supplementary Methods 1-4.

\section{Results}
\subsection{Definitions of multipartite entanglement structure}
Let us start with the definitions of multipartite entanglement structure. Considering an $N$-qubit quantum system in a Hilbert space $\mathcal{H}=\mathcal{H}_2^{\otimes N}$, one can partition the whole system into $m$ nonempty disjoint subsystems $A_i$, i.e., $\{N\}\equiv\{1,2,\dots,N\}=\bigcup_{i=1}^m A_i$ with $\mathcal{H}=\bigotimes^m_{i=1}\mathcal{H}_{A_i}$. Denote this partition to be $\mathcal{P}_m=\{A_i\}$ and omit the index $m$ when it is clear from the context. Similar to definitions of regular separable states, here, we define fully- and bi-separable states with respect to a specific partition $\mathcal{P}_m$ as follows.

\begin{definition}\label{Def:FullSep}
An $N$-qubit pure state, $\ket{\Psi_f}\in\mathcal{H}$, is $\mathcal{P}$-fully separable, iff it can be written as,
\begin{equation}\label{Eq:fullsep}
\ket{\Psi_{f}}=\bigotimes^m_{i=1}\ket{\Phi_{A_i}}.
\end{equation}
An $N$-qubit mixed state $\rho_f$ is $\mathcal{P}$-fully separable, iff it can be decomposed into a convex mixture of $\mathcal{P}$-fully separable pure states,
\begin{equation}\label{Eq:fullsepMix}
\rho_f=\sum_ip_i\ket{\Psi_f^i}\bra{\Psi_f^i},
\end{equation}
with $p_i\ge0, \forall i$ and $\sum_i p_i=1$.
\end{definition}

Denote the set of $\mathcal{P}$-fully separable states to be $S_f^{\mathcal{P}}$. Thus, if one can confirm that a state $\rho \notin S_f^{\mathcal{P}}$, the state $\rho$ should possess entanglement between the subsystems $\{A_i\}$. Such kind of entanglement could be weak though, since it only requires at least two subsystems to be entangled. For instance, the state $\ket{\Psi}=\ket{\Psi_{A_1A_2}}\otimes\prod^m_{i=3}\ket{\Psi_{A_i}}$ is called entangled nevertheless only with entanglement between $A_1$ and $A_2$. It is interesting to study the states where all the subsystems are genuinely entangled with each other. Below, we define this genuine entangled state via $\mathcal{P}$-bi-separable states.

\begin{definition}\label{Def:Bisep}
An $N$-qubit pure state, $\ket{\Psi_s}\in\mathcal{H}$, is $\mathcal{P}$-bi-separable, iff there exists a subsystem bipartition $\{A,\bar{A}\}$, where $A=\bigcup_i A_i$, $\bar{A}=\{N\}/A \neq \emptyset$, the state can be written as,
\begin{equation}\label{Eq:2sep}
\ket{\Psi_b}=\ket{\Phi_{A}}\otimes \ket{\Phi_{\bar{A}}}.
\end{equation}
An $N$-qubit mixed state $\rho_b$ is $\mathcal{P}$-bi-separable, iff it can be decomposed into a convex mixture of $\mathcal{P}$-bi-separable pure states,
\begin{equation}\label{Eq:fullsepMix}
\rho_b=\sum_ip_i\ket{\Psi_b^i}\bra{\Psi_b^i},
\end{equation}
with $p_i\ge0, \forall i$ and $\sum_i p_i=1$, and each state $\ket{\Psi_b^i}$ can have different bipartitions.
\end{definition}

Denote the set of bi-separable states to be $S_{b}^{\mathcal{P}}$. It is not hard to see that $S_{f}^{\mathcal{P}}\subset S_{b}^{\mathcal{P}}$.

\begin{definition}\label{Def:Pgenent}
A state $\rho$ possesses $\mathcal{P}$-genuine entanglement iff $\rho \notin S_{b}^{\mathcal{P}}$.
\end{definition}

\begin{figure}[tbh!]
\centering
\includegraphics[width=0.35\textwidth]{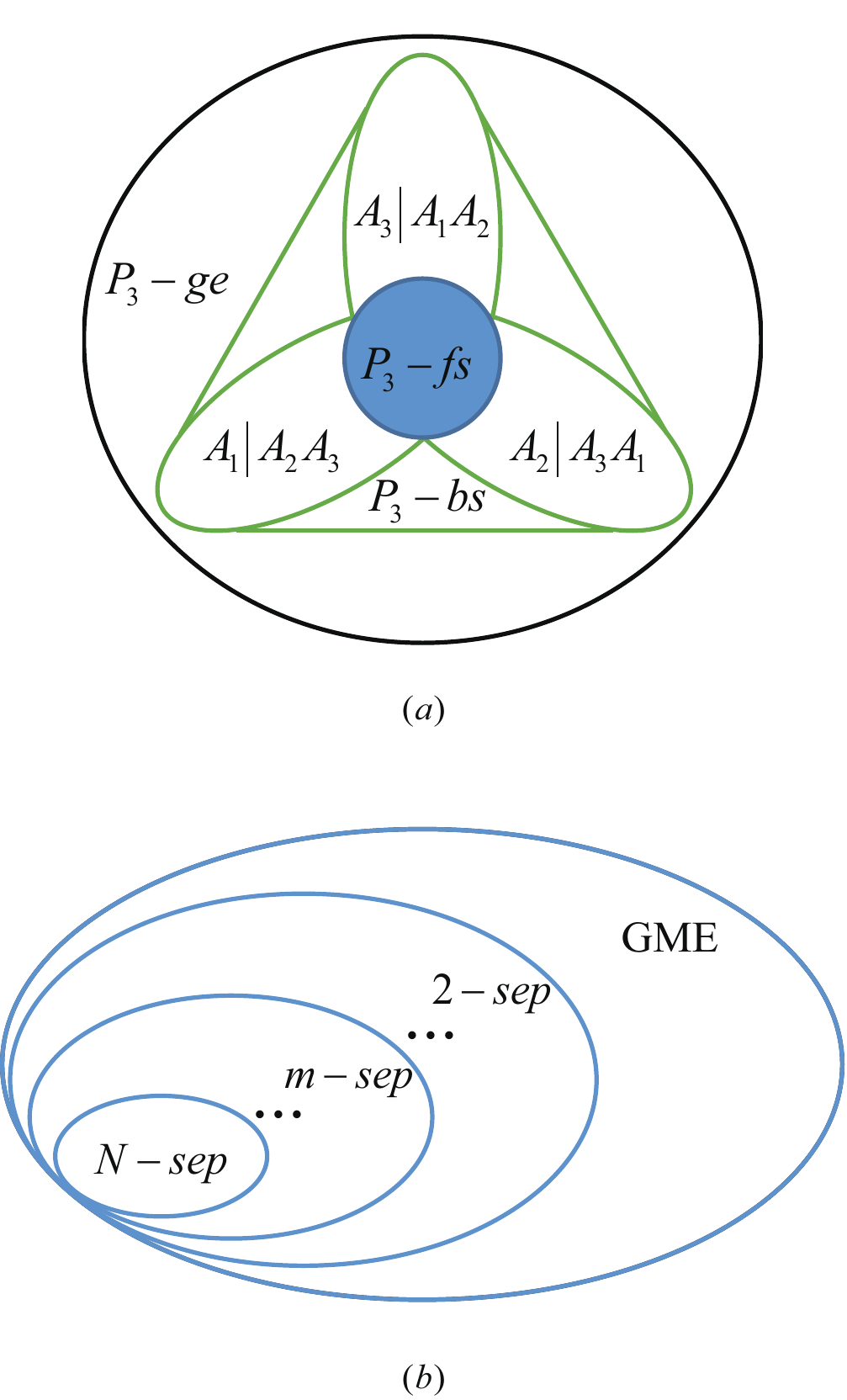}
\caption{Venn diagrams to illustrate relationships of several separable sets. (a) To illustrate the separability definitions based on a given partition, we consider a tripartition $\mathcal{P}_3=\{A_1,A_2,A_3\}$ here. The $\mathcal{P}$-fully separable state set $S_{f}^{\mathcal{P}}$ is at the centre, contained in three bi-separable sets with different bipartitions. The $\mathcal{P}$-bi-separable state set $S_{b}^{\mathcal{P}}$ is the convex hull of these three sets. A state possesses $\mathcal{P}$-genuine entanglement if it is outside of $S_{b}^{\mathcal{P}}$. Note that this becomes the case of three-qubit entanglement when each party $A_i$ contains one qubit \cite{Acin2001Classification}. (b) Separability hierarchy of $N$-qubit state with $S_{m+1}\subset S_{m}$ and $2\leq m \leq N$.  The $m$-separable state set $S_m$ is the convex hull of separable states with different $m$-partitions. Thus $S_{f}^{\mathcal{P}_m}\subset S_m$, and one can investigate $S_m$ by considering all $S_{f}^{\mathcal{P}_m}$. A state possesses genuine multipartite entanglement (GME) if it is outside of $S_2$, and is (fully) $N$-separable if it is in $S_N$.
}\label{Fig:VennD}
\end{figure}

The three entanglement-structure definitions of $\mathcal{P}$-fully separable, $\mathcal{P}$-bi-separable, and $\mathcal{P}$-genuinely entangled states can be viewed as generalized versions of regular fully separable, bi-separable, and genuinely entangled states, respectively. In fact, when $m=N$, these pairs of definitions are the same.

Following the conventional definitions, a pure state $\ket{\Psi_m}$ is $m$-separable if there exists a partition $\mathcal{P}_m$, the state can be written in the form of Eq.~\eqref{Eq:fullsep}. The $m$-separable state set, $S_m$, contains all the convex mixtures of the $m$-separable pure states, $\rho_m=\sum_ip_i\ket{\Psi_m^i}\bra{\Psi_m^i}$,  where the partition for each term $\ket{\Psi_m^i}$ needs not to be same. It is not hard to see that $S_{m+1}\subset S_{m}$. Meanwhile, define the \emph{entanglement intactness} of a state $\rho$ to be $m$, iff $\rho\notin S_{m+1}$ and $\rho\in S_{m}$.
Thus, as $\rho\notin S_{m+1}$, the intactness is at most $m$, i.e., the non-separability can serve as an upper bound of the intactness. When the entanglement intactness is 1, the state is genuinely entangled; and when the intactness is $N$, the state is fully separable. See FIG.~\ref{Fig:VennD} for the relationships among these definitions.

By definitions, one can see that if a state is $\mathcal{P}_m$-fully separable, it must be $m$-separable. Of course, an $m$-separable state might not be $\mathcal{P}_m$-fully separable, for example, if the partition is not properly chosen.
%Interestingly, there exist some $m$-separable states that are not $\mathcal{P}_m$-fully separable for any $m$-partition.
In experiment, it is important to identify the partition under which the system is fully separated. With the partition information, one can quickly identify the links where entanglement is broken.
Moreover, for some systems, such as distributed quantum computing, multiple quantum processor, and quantum network, natural partition exists due to the system geometric configuration. Therefore, it is practically interesting to study entanglement structure under partitions.

\subsection{Entanglement structure detection method}
Let us first recap the basics of graph states and the stabilizer formalism \cite{Briegel2001Persistent,Hein2006Graph}. In a graph, denoted by $G=(V,E)$, there are a vertex set $V=\{N\}$ and a edge set $E\subset[V]^2$. Two vertexes $i$, $j$ are called neighbors if there is an edge $(i,j)$ connecting them. The set of neighbors of the vertex $i$ is denoted as $N_i$. A graph state is defined on a graph $G$, where the vertexes represent the qubits initialized in the state of $\ket{+}=(\ket{0}+\ket{1})/\sqrt{2}$ and the edges represent a Controlled-$Z$ (C-$Z$) operation, $\CZ^{\{i,j\}}=\ket{0}_i\bra{0}\otimes \mathbb{I}_j+\ket{1}_i\bra{1}\otimes Z_j$, between the two neighbor qubits. Then the graph state can be written as,
\begin{equation} \label{eq:graphstate}
\ket{G}=\prod_{(i,j)\in E}\CZ^{\{i,j\}}\ket{+}^{\otimes N}.
\end{equation}
Denote the Pauli operators on qubit $i$ to be $X_i,Y_i,Z_i$. An $N$-partite graph state can also be uniquely determined by $N$ independent stabilizers,
\begin{equation}\label{}\label{Eq:Stab}
S_i=X_i\bigotimes_{j\in N_i} Z_j,
\end{equation}
which commute with each other and $S_i\ket{G}=\ket{G},\, \forall i$. That is, the graph state is the unique eigenstate with eigenvalue of $+1$ for all the $N$ stabilizers. Here, $S_i$ contains identity operators for all the qubits that do not appear in Eq.~\eqref{Eq:Stab}. As a result, a graph state can be written as a product of stabilizer projectors,
\begin{equation}\label{Eq:Gsta}
\ket{G}\bra{G}=\prod_{i=1}^N\frac{S_i+\mathbb{I}}{2}.
\end{equation}
The fidelity between $\rho$ and a graph state $\ket{G}$ can be obtained from measuring all possible products of stabilizers. However, as there are exponential terms in Eq.~\eqref{Eq:Gsta}, this process is generally inefficient for large systems. Hereafter, we consider the connected graph, since its corresponding graph state is genuinely entangled.

Now, we propose a systematic method to detect entanglement structures based on graph states. First, we give fidelity bounds between separable states and graph states as the following proposition.

\begin{proposition}\label{Th:Fidmain}
Given a graph state $\ket{G}$ and a partition $\mathcal{P}=\{A_i\}$, the fidelity between $\ket{G}$ and any $\mathcal{P}$-fully separable state is upper bounded by
\begin{equation}\label{Eq:upfullsep}
\mathrm{Tr}\left(\ket{G}\bra{G}\rho_f\right)\leq \min_{\{A,\bar{A}\}} 2^{-S(\rho_A)};
\end{equation}
and the fidelity between $\ket{G}$ and any $\mathcal{P}$-bi-separable state is upper bounded by
\begin{equation}\label{Eq:up2sep}
\mathrm{Tr}(\ket{G}\bra{G}\rho_b)\leq \max_{\{A,\bar{A}\}} 2^{-S(\rho_A)},
\end{equation}
where $\{A,\bar{A}\} $ is a bipartition of $\{A_i\}$, and $S(\rho_A)=-\mathrm{Tr}[\rho_A\log_2\rho_A]$ is the von Neumann entropy of the reduced density matrix $\rho_A=\mathrm{Tr}_{\bar{A}}(\ket{G}\bra{G})$.
\end{proposition}

The bound in Eq.~\eqref{Eq:up2sep} is tight, i.e., there always exists a $\mathcal{P}$-bi-separable state to saturate it. The bound in Eq.~\eqref{Eq:upfullsep} may not be tight for some partition $\mathcal{P}=\{A_i\}$ and some graph state $\ket{G}$. In addition, we remark that to extend Theorem \ref{Th:Fidmain} from the graph state to a general state $\ket{\Psi}$, one should substitute the entropy in the bounds of Eqs.~\eqref{Eq:upfullsep} and \eqref{Eq:up2sep} with the min-entropy $S_{\infty}(\rho_A)=-\log \lambda_1$ with $\lambda_1$ the largest eigenvalue of $\rho_A=\mathrm{Tr}_{\bar{A}}(\ket{\Psi}\bra{\Psi})$.
%Given a bipartition, the entanglement entropy $S(\rho_A)$ relates to the rank of the adjacency matrix of the subgraph $G_{A\bar{A}}$. The detailed calculation can be found in Supplemental Material. However, the optimization problems in Eq.~\eqref{Eq:upfullsep} and \eqref{Eq:up2sep} are computationally hard in general due to the exponential number of possible bipartitions, one can solve it properly as the number of the subsystems $m$ is not too large. In addition, we can always have an upper bound on the minimisation by only considering specific partitions. Analytical calculation of the optimization is possible for graph states with certain symmetries, such as the 1 (2)-D cluster state and the GHZ state, as we will discuss later in Corollaries.

Next, we propose an efficient method to lower bound the fidelity between an unknown prepared state and the target graph state. A graph is $k$-colorable if one can divide the vertex set into $k$ disjoint subsets $\bigcup V_l=V$ such that any two vertexes in the same subset are not connected. The smallest number $k$ is called the chromatic number of the graph \footnote{Note that the colorability is a property of the graph (not the state), one may reduce the number of measurement settings by local Clifford operations \cite{Hein2006Graph}.}.
We define the stabilizer projector of each subset $V_l$ as
\begin{equation}\label{Eq:Pi}
\begin{aligned}
P_l=\prod_{i\in V_l}  \frac{S_i+\mathbb{I}}{2},
\end{aligned}
\end{equation}
where $S_i$ is the stabilizer of $\ket{G}$ in subset $V_l$. The expectation value of each $P_l$ can be obtained by one local measurement setting $\bigotimes_{i\in V_l} X_i \bigotimes_{j\in V/V_l} Z_j$. Then, we can propose a fidelity evaluation scheme with $k$ local measurement settings, as the following proposition.

\begin{proposition}\label{Th:lowerbound}
For a graph state $\ket{G}\bra{G}$ and the projectors $P_l$ defined in Eq.~\eqref{Eq:Pi}, the following inequality holds,
\begin{equation}\label{Eq:FidLB}
\ket{G}\bra{G}\ge\sum_{l=1}^k P_l-(k-1)\mathbb{I},
\end{equation}
where $A\geq B$ indicates that $(A-B)$ is positive semidefinite.
\end{proposition}

Note that Proposition \ref{Th:lowerbound} with $k=2$ case has also been studied in literature \cite{Toth2005Detecting}. Combining Propositions~\ref{Th:Fidmain} and \ref{Th:lowerbound}, we propose entanglement structure witnesses with $k$ local measurement settings, as presented in the following theorem.

\begin{theorem}\label{Th:main}
Given a partition $\mathcal{P}=\{A_i\}$, the operator $W_{f}^{\mathcal{P}}$ can witness non-$\mathcal{P}$-fully separability (entanglement),
\begin{equation}\label{Eq:witen}
W_{f}^{\mathcal{P}} = \left(k-1+\min_{\{A,\bar{A}\}} 2^{-S(\rho_A)}\right)\mathbb{I}-\sum_{l=1}^k P_l,
\end{equation}
with $\langle W_{f}^{\mathcal{P}}\rangle \ge0$ for all $\mathcal{P}$-fully-separable states; and the operator $W_{b}^{\mathcal{P}}$ can witness $\mathcal{P}$-genuine entanglement,
\begin{equation}\label{Eq:witGen}
W_{b}^{\mathcal{P}} = \left(k-1+\max_{\{A,\bar{A}\}} 2^{-S(\rho_A)}\right)\mathbb{I}-\sum_{l=1}^k P_l,
\end{equation}
with $\langle W_{b}^{\mathcal{P}}\rangle\ge0$ for all $\mathcal{P}$-bi-separable states, where $\{A,\bar{A}\} $ is a bipartition of $\{A_i\}$, $\rho_A=\mathrm{Tr}_{\bar{A}}(\ket{G}\bra{G})$, and the projectors $P_l$ is defined in Eq.~\eqref{Eq:Pi}.
\end{theorem}
The proposed entanglement structure witnesses have several favourable features. First, given an underlying graph state, the implementation of the witnesses is the same for different partitions. This feature allows us to study different entanglement structures in one experiment. Note that the witness operators in Eqs.~\eqref{Eq:witen} and \eqref{Eq:witGen} can be divided into two parts: The measurement results of $P_l$ obtained from the experiment rely on the prepared unknown state and are independent of the partition; The bounds, $1+\min\,(\max)_{\{A,\bar{A}\}}2^{-S(\rho_A)}$, on the other hand, rely on the partition and are independent of the experiment. Hence, in the data postprocessing of the measurement results of $P_l$, we can study various entanglement structures for different partitions by calculating the corresponding bounds analytically or numerically.

Second, besides investigating the entanglement structure among all the subsystems, one can also employ the same experimental setting to study that of a subset of the subsystems, by performing different data post-processing. For example, suppose the graph $G$ is partitioned into 3 parts, say $A_1$, $A_2$ and $A_3$, and only the entanglement between subsystems $A_1$ and $A_2$ is of interest. One can construct new witness operators with  projectors $P_l'$, by replacing all the Pauli operators on the qubits in $A_3$ in Eq.~\eqref{Eq:Pi} to identities. Such measurement results can be obtained by processing the measurement results of the original $P_l$. Then the entanglement between $A_1$ and $A_2$ can be detected via Theorem \ref{Th:main} with projectors $P_l'$ and the corresponding bounds of the graph state $\ket{G_{A_1A_2}}$. Details are discussed in Supplementary Notes 1.

Third, when each subsystem $A_i$ contains only one qubit, that is, $m=N$, the witnesses in Theorem \ref{Th:main} become the conventional ones. It turns out that Eq.~\eqref{Eq:witGen} is the same for all the graph states under the $N$-partition $\mathcal{P}_N$, as shown in the following corollary. Note that, a special case of the corollary, the genuine entanglement witness for the GHZ and 1-D cluster states, has been studied in literature \cite{Toth2005Detecting}.

\begin{corollary} \label{Th:GHZ}
The operator $W_{b}^{\mathcal{P}_N}$ can witness genuine multipartite entanglement,
\begin{equation}\label{Eq:GHZgenuine}
W_{b}^{\mathcal{P}_N} = \left(k-\frac12\right)\mathbb{I}-\sum_{l=1}^k P_l,
\end{equation}
with $\langle W_{b}^{\mathcal{P}_N}\rangle \ge0$ for all bi-separable states, where $P_l$ is defined in Eq.~\eqref{Eq:Pi} for any graph state.
\end{corollary}

Fourth, the witness in Eq.~\eqref{Eq:witen} can be naturally extended to identify  non-$m$-separability, by investigating all possible partitions $\mathcal{P}_m$ with fixed $m$. In fact, according to the definition of $m$-separable states and Eq.~\eqref{Eq:upfullsep}, the fidelity between any $m$-separable state $\rho_m$ and the graph state $\ket{G}$ can be upper bounded by $\max_{\mathcal{P}_m}\min_{\{A,\bar{A}\}} 2^{-S(\rho_A)}$, where the maximization is over all possible partitions with $m$ subsystems. As a result, we have the following theorem on the non-$m$-separability.

\begin{theorem}\label{Th:msep}
The operator $W_m$ can witness non-$m$-separability,
\begin{equation}\label{Eq:msep}
W_m = \left(k-1+\max_{\mathcal{P}_m}\min_{\{A,\bar{A}\}} 2^{-S(\rho_A)}\right)\mathbb{I}-\sum_{l=1}^k P_l,
\end{equation}
with $\langle W_m\rangle\ge0$ for all $m$-separable states, where the maximization takes over all possible partitions $\mathcal{P}_m$ with $m$ subsystems, the minimization takes over all bipartition of $\mathcal{P}_m$, $\rho_A=\mathrm{Tr}_{\bar{A}}(\ket{G}\bra{G})$, and the projector $P_l$ is defined in Eq.~\eqref{Eq:Pi}.
\end{theorem}
The robustness analysis of the witnesses proposed in Theorem \ref{Th:main} and \ref{Th:msep} under the white noise is presented in Methods. It shows that our entanglement structure witnesses are quite robust to noise. Moreover, the optimization problems in Theorem \ref{Th:main} and \ref{Th:msep} are generally hard, since there are exponentially many different possible partitions.
%The optimization problem in Theorem \ref{Th:msep} is generally hard, since there are about ${m^N}/{m!}$ possible ways to partition $N$ qubits into $m$ subsystems. For example, when $N$ is large (say, in the order of 70 qubits), the number of different partitions is exponentially large even with a small separability number $m$. Similarly, the optimization problem in Theorem \ref{Th:main} is also generally hard.
Surprisingly, for several widely-used types of graph states, such as 1-D, 2-D cluster states, and the GHZ state, we find the analytical solutions to the optimization problem, as shown in the following section.

\subsection{Applications to several typical graph states}
In this section, we apply the general entanglement detection method proposed above to several typical graph states, 1-D, 2-D cluster states, and the GHZ state. Note that for these states the corresponding graphs are all 2-colorable. Thus, we can realize the witnesses with only two local measurement settings. For clearness, the vertexes in the subsets $V_1$ and $V_2$ are associated with red and blue colors respectively, as shown in FIG.~\ref{Fig:allstate}. We write the stabilizer projectors defined in Eq.~\eqref{Eq:Pi} for the two subsets as,
\begin{equation}\label{Eq:P1P2}
\begin{aligned}
P_1 &= \prod_{red\ i}  \frac{S_i+\mathbb{I}}{2}, \\
P_2 &= \prod_{blue\ i} \frac{S_i+\mathbb{I}}{2}.
\end{aligned}
\end{equation}
The more general case with $k$-chromatic graph states is presented in Supplementary Notes 1.

\begin{figure}[tbh]
\centering
\includegraphics[width=0.45\textwidth]{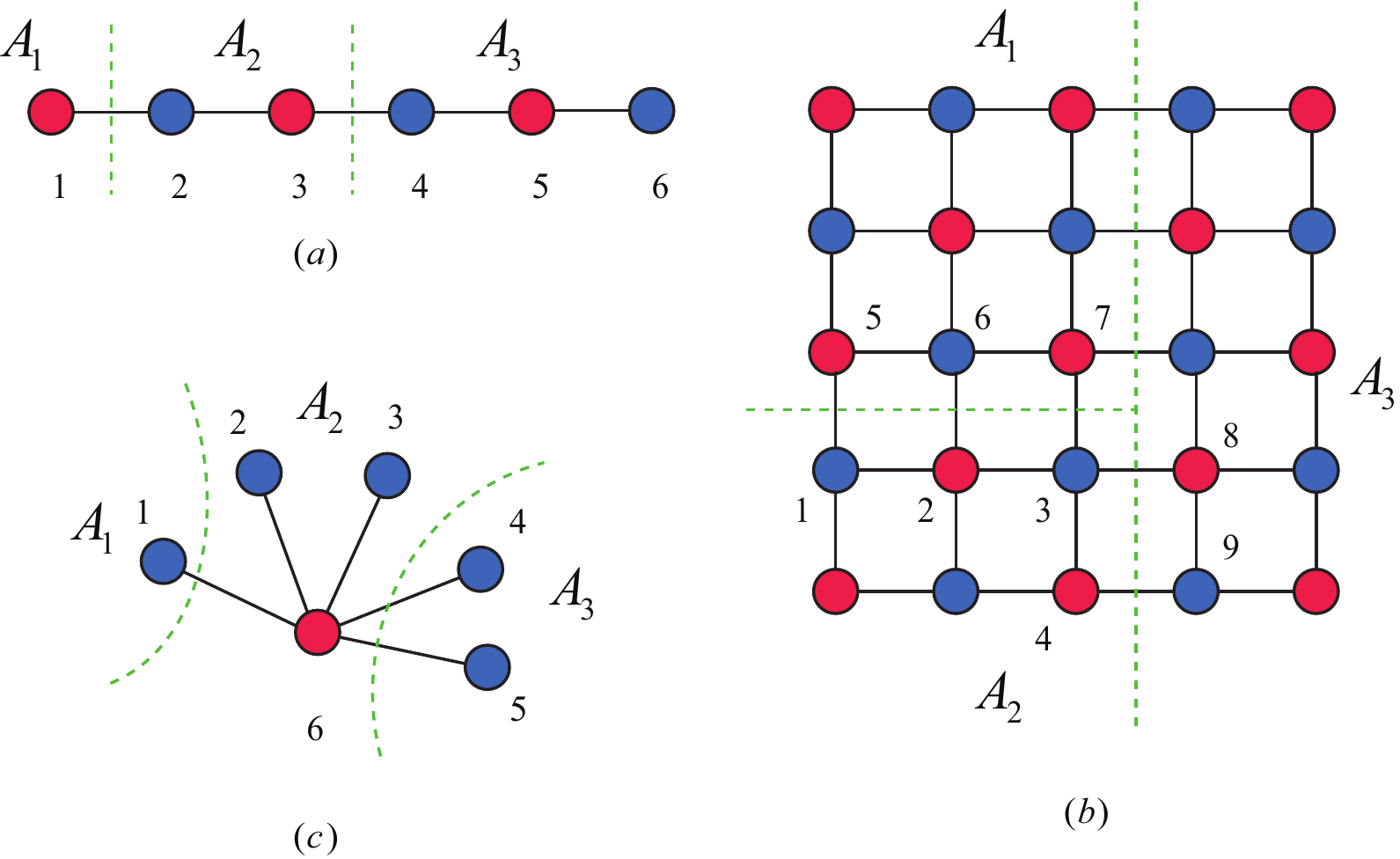}
\caption{Graphs of the (a) 1-D cluster state $\ket{C_1}$, (b) 2-D cluster state $\ket{C_2}$, and (c) GHZ state $\ket{GHZ}$. Note that the graph state form of the GHZ state is equivalent to its canonical form, $(\ket{0}^{\otimes N}+\ket{1}^{\otimes N})/\sqrt{2}$, up to local unitary operations.
}\label{Fig:allstate}
\end{figure}

We start with a 1-D cluster state $\ket{C_1}$ with stabilizer projectors in the form of Eq.~\eqref{Eq:P1P2}. Consider an example of tripartition $\mathcal{P}_3=\{A_1,A_2,A_3\}$, as shown in FIG.~\ref{Fig:allstate}(a), there are three ways to divide the three subsystems into two sets, i.e., $\{A,\bar{A}\}=\{A_1, A_2A_3\}$, $\{A_2, A_1A_3\}$, $\{A_3, A_1A_2\}$. It is not hard to see that the corresponding entanglement entropies are $S(\rho_{A_1})=S(\rho_{A_3})=1$ and $S(\rho_{A_2})=2$. Note that in the calculation, each broken edge will contribute 1 to the entropy, which is a manifest of the area law of entanglement entropy \cite{Eisert2010area}. According to Theorem \ref{Th:main}, the operators to witness $\mathcal{P}_3$-entanglement structure are given by,
\begin{equation}\label{Eq:witen1D}
\begin{aligned}
W_{f,C_1}^{\mathcal{P}_3} &= \frac{5}{4}\mathbb{I}-(P_1+P_2), \\
W_{b,C_1}^{\mathcal{P}_3} &= \frac{3}{2}\mathbb{I}-(P_1+P_2),
\end{aligned}
\end{equation}
where the two projectors $P_1$ and $P_2$ are defined in Eq.~\eqref{Eq:P1P2} with the graph of FIG.~\ref{Fig:allstate}(a). %The example of 2-D cluster state $\ket{C_2}$
%%defined in a $5\times5$ lattice under a tripartition
%shown in Fig.~\ref{Fig:allstate}(b), can be treated in a similar manner. Details are presented in Supplemental Material.

Next, we take an example of 2-D cluster state $\ket{C_2}$ defined in a $5\times5$ lattice and consider a tripartition, as shown in FIG.~\ref{Fig:allstate}(b). Similar to the 1-D cluster state case with the area law, the corresponding entanglement entropies are $S(\rho_{A_1})=S(\rho_{A_3})=5$ and $S(\rho_{A_2})=4$. According to Theorem \ref{Th:main}, the operators to witness $\mathcal{P}_3$-entanglement structure are given by,
\begin{equation}\label{Eq:witen2D}
\begin{aligned}
W_{f,C_2}^{\mathcal{P}_3} &= \frac{33}{32}\mathbb{I}-(P_1+P_2), \\
W_{b,C_2}^{\mathcal{P}_3} &= \frac{17}{16}\mathbb{I}-(P_1+P_2),
\end{aligned}
\end{equation}
where the two projectors $P_1$ and $P_2$ are defined in Eq.~\eqref{Eq:P1P2} with the graph of FIG.~\ref{Fig:allstate}(b). Similar analysis works for other partitions and other graph states.

Now, we consider the case where each subsystem $A_i$ contains exactly one qubit, $\mathcal{P}_N$. Then, witnesses in Eq.~\eqref{Eq:witen} become the conventional ones, as shown in the following Corollary.

\begin{corollary}\label{Th:1DN}
The operator $W_{f,C}^{\mathcal{P}_N}$ can witness non-fully separability (entanglement),
\begin{equation}\label{Eq:1Dent}
W_{f,C}^{\mathcal{P}_N} = (1+2^{-\lfloor\frac{N}{2}\rfloor})\mathbb{I}-(P_1+P_2),
\end{equation}
with $\langle W_{f,C}^{\mathcal{P}_N}\rangle \ge0$ for all fully separable states, where the two projectors $P_1$ and $P_2$ are defined in Eq.~\eqref{Eq:P1P2} with the stabilizers of any 1-D or 2-D cluster state.
\end{corollary}

Here, we only show the cases of 1-D and 2-D cluster states. We conjecture that the witness holds for any (such as 3-D) cluster states. For a general graph state, on the other hand, the corollary does not hold. In fact, we have a counter example of the GHZ state shown in FIG.~\ref{Fig:allstate}(c). It is not hard to see that for any GHZ state, the entanglement entropy is given by,
\begin{equation}\label{Eq:SrhoGHZ}
\begin{aligned}
S(\rho_A^{GHZ}) = 1, \;\;\; \forall \{A,\bar{A}\}.
\end{aligned}
\end{equation}
Then, Eqs.~\eqref{Eq:witen} and \eqref{Eq:witGen} yield the same witnesses. That is, the witness constructed by Theorem \ref{Th:main} for the GHZ state can only tell genuine entanglement or not.

%Following Theorem \ref{Th:msep}, one can fix the number of the subsystems $m$ and investigate all possible partitions to detect the non-$m$-separability. The optimization problem can be solved analytically for the 1-D and 2-D cluster states, as shown in Corollary \ref{Th:1Dmsep} and \ref{Th:2Dmsep}, respectively.

Following Theorem \ref{Th:msep}, one can fix the number of the subsystems $m$ and investigate all possible partitions to detect the non-$m$-separability. The optimization problem can be solved analytically for the 1-D and 2-D cluster states, as shown in Corollary \ref{Th:1Dmsep} and \ref{Th:2Dmsep}, respectively.

\begin{corollary}\label{Th:1Dmsep}
The operator $W_{m,C_1}$ can witness non-$m$-separability,
\begin{equation}\label{Eq:1Dmsep}
W_{m,C_1} = (1+2^{-\lfloor\frac{m}{2}\rfloor})\mathbb{I}-(P_1+P_2),
\end{equation}
with $\langle W_{m,C_1} \rangle\ge0$ for all $m$-separable states, where the two projectors $P_1$ and $P_2$ are defined in Eq.~\eqref{Eq:P1P2} with the stabilizers of a 1-D cluster state.
\end{corollary}

In particular, when $m=2$ and $m=N$, $W_{m,C_1}$ becomes the entanglement witnesses in  Eqs.~\eqref{Eq:GHZgenuine} and \eqref{Eq:1Dent}, respectively. %Comparing to the 1-D cluster state case, the partition of the 2-D cluster state is much richer and its separability witness is shown in the following corollary.

\begin{corollary}\label{Th:2Dmsep}
The operator $W_{m,C_2}$ can witness non-$m$-separability for $N\geq m(m-1)/2$,
\begin{equation}\label{Eq:2Dmsep}
W_{m,C_2} = \left(1+2^{-\left\lceil \frac{-1+\sqrt{1+8(m-1)}}{2}\right\rceil}\right)\mathbb{I}-(P_1+P_2),
\end{equation}
with $\langle W_{m,C_2} \rangle\ge0$ for all $m$-separable states, where the two projectors $P_1$ and $P_2$ are defined in Eq.~\eqref{Eq:P1P2} with the stabilizers of a 2-D cluster state.
\end{corollary}

%Here we consider the separability parameter $m\leq 5$, since compared with the 1-D case the optimization of the entanglement entropy is hard. However, the parameter regime we consider here is sufficient for the experiment detection, as it only needs the expectation value $\langle P_1+P_2\rangle > 1.125$ for $m=5$. Roughly speaking, the requirement for the experiment is to prepare a state whose fidelity to the 2-D cluster state is larger than $0.125$.

We remark that the witnesses constructed in Corollaries \ref{Th:GHZ}, \ref{Th:1DN}, and \ref{Th:1Dmsep} are tight. Take the witness $W_{m,C_1}$ in Corollary \ref{Th:1Dmsep} as an example. There exists an $m$-separable state $\rho_m$ that saturates $\mathrm{Tr}(\rho_m W_{m,C_1})=0$. In addition, as $m\leq 5$, the witness $W_{m,C_2}$ in Corollary \ref{Th:2Dmsep} is also tight. Detailed discussions are presented in Supplementary Methods 1-4.

\section{Discussion}
In this work, we propose a systematic method to construct efficient witnesses to detect entanglement structures based on graph states.
Our method offers a standard tool for entanglement structure detection and multipartite quantum system benchmarking. The entanglement structure definitions and the associated witness method may further help to detect novel quantum phases, by investigating the entanglement properties of the ground states of related Hamiltonians \cite{Bei2019Meets}.

The witnesses proposed in this work can be directly generalized to stabilizer states \cite{gottesman1997stabilizer,Nielsen2011Quantum}, which are equivalent to graph states up to local Clifford operations \cite{Hein2006Graph}. It is interesting to extend the method to more general multipartite quantum states, such as the hyper-graph state \cite{Rossi2013hyper} and the tensor network state \cite{ORUS2014TNW}. Meanwhile, the generalization to the neural network state \cite{Carleo2017Solving} is also intriguing, since this kind of ansatz is able to represent broader quantum states with a volume law of entanglement entropy \cite{Deng2017Neural}, and is a fundamental block for potential artificial intelligence applications. In addition, one may utilize the proposed witness method to detect other multipartite entanglement properties, such as the entanglement depth and width \cite{Sorensen2001depth,Wolk2016width}, as $m$-separability in this work.  Moreover, one can also consider the self-testing scenario, such as (measurement-) device-independent settings \cite{Branciard2013Measurement,Liang2015Bell,Zhao2016Efficient}, which can help to manifest the entanglement structures with less assumptions on the devices. Furthermore, translating the proposed entanglement witnesses into a probabilistic scheme is also interesting \cite{Dimic2018Single,Saggio2019few}.

%Finally, considering the case of GHZ state, all the entanglement witnesses are reduced to one operator, which hinders us from investigating more detailed entanglement structures. Thus, it is also necessary to develop other methods to resolve this situation, such as the witnesses proposed in Ref.~\cite{Lu2018Structure} for GHZ-like states.

\section{Methods}\label{Sec:Method}
\subsection{Proof of Proposition \ref{Th:Fidmain}}\label{Sec:P:one}
\begin{proof}
First, let us prove the $\mathcal{P}$-bi-separable state case in Eq.~\eqref{Eq:up2sep}.
Since the $\mathcal{P}$-bi-separable state set $S_b^{\mathcal{P}}$ is convex, one only needs to consider the fidelity $|\bra{\Psi_b}G\rangle|^2$ of the pure state $\ket{\Psi_b}$ defined in Eq.~\eqref{Eq:2sep}.
It is known that the maximal value of the fidelity equals to the largest Schmidt coefficient of $\ket{G}$ with regard to the bipartition $\{A,\bar{A}\}$ \cite{Bourennane2004Experimental}, i.e.,
\begin{equation}\label{}
\max_{\ket{\Psi_b}}|\bra{\Psi_b}G\rangle|^2=\lambda_1,
\end{equation}
with the Schmidt decomposition  $\ket{G}=\sum_{i=1}^d\sqrt{\lambda_i}\ket{\Phi_i}_{A}\ket{\Phi'_i}_{\bar{A}}$ and $\lambda_1\geq \lambda_2\geq\cdots \geq\lambda_d$.
For general graph state $\ket{G}$, the spectrum of any reduced density matrix $\rho_A$ is flat, i.e., $\lambda_1=\lambda_2=\cdots\lambda_d$, with $d$ being the rank of $\rho_A$ \cite{Hein2004Multiparty}. As a result, one has
\begin{equation}\label{}
\begin{aligned}
S(\rho_A)=\log_2 d, \\
\lambda_i=\frac{1}{d}=2^{-S(\rho_A)}.
\end{aligned}
\end{equation}
To get an upper bound, one should maximize $2^{-S(\rho_A)}$ on all possible subsystem bipartitions and then get Eq.~\eqref{Eq:up2sep}.

Second, we prove the $\mathcal{P}$-fully separable state case in Eq.~\eqref{Eq:upfullsep}.
Similarly, we only need to upper bound the fidelity of the pure state $\ket{\Psi_{f}}$ defined in Eq.~\eqref{Eq:fullsep}, due to the convexity property of the $\mathcal{P}$-fully separable state set $S_{f}^{\mathcal{P}}$.  From the proof of Eq.~\eqref{Eq:up2sep} above, we know that the fidelity of the $\mathcal{P}$-bi-separable state satisfies the bound $|\bra{\Psi_b}G\rangle|^2\leq 2^{-S(\rho_A)}$, given a subsystem bipartition $\{A,\bar{A}\}$. It is not hard to see that these bounds all hold for $\ket{\Psi_{f}}$, since $S_{f}^{\mathcal{P}}\subset S_{b}^{\mathcal{P}}$. Thus, one can obtain the finest bound via minimizing over all possible bipartitions and finally get Eq.~\eqref{Eq:upfullsep}.
\end{proof}

The entanglement entropy $S(\rho_A)$ equals the rank of the adjacency matrix of the underlying bipartite graph, which can be efficiently calculated. Details are discussed in Supplementary Notes 1. While the optimization problems can be computationally hard due to the exponential number of possible bipartitions, one can solve it properly as the number of the subsystems $m$ is not too large. In addition, we can always have an upper bound on the minimization by only considering specific partitions. Analytical calculation of the optimization is possible for graph states with certain symmetries, such as the 1-D and 2-D cluster states and the GHZ state.
\subsection{Proof of Proposition \ref{Th:lowerbound}}\label{Sec:P:two}
\begin{proof}
As shown in Main Text, a graph state $\ket{G}$ can be written in the following form
\begin{equation}\label{}
\ket{G}\bra{G}=\prod_{i=1}^N\frac{S_i+\mathbb{I}}{2}=\prod_{l=1}^k P_l.
\end{equation}
Accordingly, Eq.~\eqref{Eq:FidLB} in Proposition \ref{Th:lowerbound} becomes,
\begin{equation}\label{Eq:FidLB1}
\left[\prod_{l=1}^k P_l+(k-1)\mathbb{I}\right]-\sum_{l=1}^k P_l\ge 0.
\end{equation}

Note that the projectors $P_l$ commute with each other, thus we can prove Eq.~\eqref{Eq:FidLB1} for all subspaces which are determined by the eigenvalues of all $P_l$. For the subspace where the eigenvalues of all $P_l$ are $1$, the inequality $(1+k-1)-k\geq 0$ holds. For the subspace where only one of $P_l$ takes value of $0$, the inequality $(0+k-1)-(k-1)\geq 0$ holds. Moreover, for the subspace in which there are more than one $P_l$ taking $0$, the inequality also holds. As a result, we finish the proof.
\end{proof}

\subsection{Proofs of Theorem \ref{Th:main} and Theorem \ref{Th:msep}}\label{Sec:Th:onetwo}
Proof of Theorem \ref{Th:main}:
\begin{proof}
The proof is to combine Proposition \ref{Th:Fidmain} and \ref{Th:lowerbound}. Here we only show the proof of Eq.~\eqref{Eq:witen}, and one can prove Eq.~\eqref{Eq:witGen} in a similar way.
To be specific, one needs to show that any $\mathcal{P}$-fully separable state satisfies $\langle W_{f}^{\mathcal{P}}\rangle \ge0$, that is,
\begin{equation}\label{Eq:profull}
\begin{aligned}
\mathrm{Tr}\left\{\sum_{l=1}^k P_l\rho_f\right\}&\leq \mathrm{Tr}\left\{\left[(k-1)\mathbb{I}+\ket{G}\bra{G}\right]\rho_f\right\}\\
&\leq (k-1)+\min_{\{A,\bar{A}\}} 2^{-S(\rho_A)}.
\end{aligned}
\end{equation}
Here the first and the second inequalities are right on account of Proposition \ref{Th:lowerbound} and  \ref{Th:Fidmain}, respectively.
\end{proof}

Proof of Theorem \ref{Th:msep}:
\begin{proof}
With Eq.~\eqref{Eq:upfullsep} one can bound the fidelity from any $\mathcal{P}$-fully separable state to a graph state $\ket{G}$. The $m$-separable state set $S_m$ contains all the state $\rho_m$ which can be written as the convex mixture of pure $m$-separable state, $\rho_m=\sum_ip_i\ket{\Psi_m^i}\bra{\Psi_m^i}$,  where the partition for each constitute $\ket{\Psi_m^i}$ needs not to be the same. Hence one can bound the fidelity from $\rho_m$ to a graph state $\ket{G}$ by investigating all possible partitions, i.e.,
\begin{equation}\label{Eq:upmsep}
\mathrm{Tr}(\ket{G}\bra{G}\rho_m)\leq \max_{\mathcal{P}_m}\min_{\{A,\bar{A}\}} 2^{-S(\rho_A)},
\end{equation}
where the maximization takes over all possible partitions $\mathcal{P}_m$ with $m$ subsystems, the minimization takes over all bipartition of $\mathcal{P}_m$. Then like in Eq.~\eqref{Eq:profull}, by combing Eqs.~\eqref{Eq:FidLB} and \eqref{Eq:upmsep} we finish the proof.
\end{proof}

The optimization problem in Theorem \ref{Th:msep} over the partitions is generally hard, since there are about ${m^N}/{m!}$ possible ways to partition $N$ qubits into $m$ subsystems. For example, when $N$ is large (say, in the order of 70 qubits), the number of different partitions is exponentially large even with a small separability number $m$.
%Similarly, as discussed below Proposition \ref{Th:Fidmain}, the optimization problem in Theorem \ref{Th:main} is also generally hard.
Surprisingly, for several widely-used types of graph states, such as 1-D, 2-D cluster states, and the GHZ state, we find the analytical solutions to the optimization problem, as shown in Corollaries in Main Text.
\subsection{Robustness of entanglement structure witnesses}\label{Sec:robust}
In this section, we discuss the robustness of the proposed witnesses in Theorem \ref{Th:main} and \ref{Th:msep}. In practical experiments, the prepared state $\rho$ deviates from the target graph state $\ket{G}$ due to some nonnegligible noise. Here we utilize the following white noise model to quantify the robustness of the witnesses.
\begin{equation}\label{Eq:whitenoise}
\rho=(1-p_{noise})\ket{G}\bra{G}+p_{noise}\frac{\mathbb{I}}{2^N},
\end{equation}
which is a mixture of the original state $\ket{G}$ and the maximally mixed state with coefficient $p_{noise}$. We will find the largest $p_{limit}$, such that the witness can detect the corresponding entanglement structure when $p_{noise}<p_{limit}$. Thus $p_{limit}$ reflects the robustness of the witness.

Let us first consider the entanglement witness $W_{f}^{\mathcal{P}}$ in Eq.~\eqref{Eq:witen} of Theorem \ref{Th:main}. For clearness, we denote $C_{min}=\min_{\{A,\bar{A}\}} 2^{-S(\rho_A)}$. Insert the state of Eq.~\eqref{Eq:whitenoise} into the witness, one gets,
\begin{equation}\label{Eq:tolFull}
\begin{aligned}
\mathrm{Tr}(W_{f}^{\mathcal{P}}\rho)&=\mathrm{Tr}\Bigg{\{}\left[\left(k-1+C_{min}\right)\mathbb{I}-\sum_{l=1}^k P_l \right]\\
&\ \ \ \ \ \ \ \ \ \left[p_{noise}\frac{\mathbb{I}}{2^N}+(1-p_{noise})\ket{G}\bra{G}\right]\Bigg{\}}\\
&=p_{noise}\left(k-1+C_{min}-2^{-N}\sum_{l=1}^k 2^{N-n_l}\right)+\\
&\ \ \ \ \ \ \ \ \ \ \ \ \ (1-p_{noise})(k-1+C_{min}-k)\\
&=p_{noise}\left(k-\sum_{l=1}^k 2^{-n_l}\right)+(C_{min}-1),
\end{aligned}
\end{equation}
where $n_l=|V_l|$ is the qubit number in each vertex set with different color, and in the second equality we employ the facts that $\mathrm{Tr}(P_l)=2^{N-n_l}$ and $\mathrm{Tr}(P_l\ket{G}\bra{G})=1$. Let the above expectation value less than zero, one has
\begin{equation}\label{Eq:tolFullnew}
p_{noise}<\frac{1-C_{min}}{k-\sum_{l=1}^k 2^{-n_l}}.
\end{equation}

Similarly, for the $\mathcal{P}$-genuine entanglement witness and the non-m-separability witness in Eqs.~\eqref{Eq:witGen} and \eqref{Eq:msep}, we have,
\begin{equation}\label{Eq:tolGennew}
\begin{aligned}
p_{noise} &< \frac{1-C_{max}}{k-\sum_{l=1}^k 2^{-n_l}}, \\
p_{noise} &< \frac{1-C_{m}}{k-\sum_{l=1}^k 2^{-n_l}},
\end{aligned}
\end{equation}
where we denote the optimizations $\max_{\{A,\bar{A}\}} 2^{-S(\rho_A)}$ and $\max_{\mathcal{P}_m}\min_{\{A,\bar{A}\}} 2^{-S(\rho_A)}$ as $C_{max}$ and $C_{m}$ respectively.

Moreover, it is not hard to see that all the coefficients $C_{min}$, $C_{max}$ and $C_{m}$ are not larger than $0.5$. Thus, for any entanglement structure witness, one has
\begin{equation}\label{Eq:tolbound}
p_{limit}\geq \frac{0.5}{k-\sum_{l=1}^k 2^{-n_l}}
> \frac{1}{2k}.
\end{equation}
As a result, our entanglement structure witness is quite robust to noise, since the largest noise tolerance $p_{limit}$ is just related to the chromatic number of the graph.

%\section{CODE AVAILABILITY}
%Code sharing is not applicable to this article as no code was generated or analyzed
%during the current study.
%
%\section{DATA AVAILABILITY}
%Data sharing is not applicable to this article as no data sets were generated or analyzed
%during the current study.

\section{Acknowledgments}
We acknowledge Y.-C.~Liang for the insightful discussions. This work was supported by the National Natural Science Foundation of China Grants No.~11875173 and No.~11674193, and the National Key R\&D Program of China Grants No.~2017YFA0303900 and No.~2017YFA0304004, and the Zhongguancun Haihua Institute for Frontier Information Technology. Xiao Yuan was supported by the EPSRC National Quantum Technology Hub in Networked Quantum Information Technology (EP/M013243/1).

%\section{COMPETING INTERESTS}
%The authors declare that there are no competing interests.

%\section{AUTHOR CONTRIBUTIONS}
%Y.~Z.~and X.~M.~initialized the project. Y.~Z., Q.~Z., and X.~Y.~developed the idea and formulated the problem as it is presented. X.~M.~supervised the project. All authors contributed to deriving the results and writing the manuscript.

\onecolumngrid

\appendix
\section*{The structure of Supplementary Material}\label{Sec:general}
The structure of Supplementary Material is organized as follows. In Supplementary Notes 1, we apply Theorem 1 in Main Text to several graph states, such as 1-D and 2-D cluster states, and prove the entanglement structure witnesses shown in Main Text. The advantage of witnessing subsystem entanglement structures and the generalization to multi-color graph states are also contained in Supplementary Notes 1.  We put the proofs of Corollaries \ref{Th:GHZ} to \ref{Th:2Dmsep} in Supplementary Methods 1 to Supplementary Methods 4, respectively. We also discuss the tightness of each witness at the end of each corollary.

\section*{SUPPLEMENTARY NOTES 1: Witness entanglement structures of graph states}\label{Sec:general}
In this section, to illustrate the proposed entanglement structure witnesses, we apply Theorem 1 to several widely-used graph states, such as 1-D and 2-D cluster states and prove the results shown in Main Text. In addition, we also discuss the advantage of witnessing subsystem entanglement structures by only post-processing the measurement results, and the generalization to multi-color graph states.
\subsection*{Entanglement entropy of graph state}\label{subsec:entfor}
Here, we briefly review the formula of the entanglement entropy of graph state \cite{Hein2004Multiparty}, which is helpful for the following discussions.

Any simple graph $G$ can be uniquely determined by its symmetric adjacency matrix denoted as $\Gamma$, with $\Gamma_{i,j}=1$ iff $(i,j)\in E$. Suppose the vertex set $V=\{N\}$ is partitioned into two complementary subsets $A$ and $\bar{A}$, the adjacency matrix $\Gamma$ can be arranged in the following form,
\begin{equation}
 \Gamma_G=\left(
 \begin{array}{cc}
    \Gamma_A & \Gamma_{A\bar{A}}\\
    \Gamma_{A\bar{A}}^T & \Gamma_{\bar{A}}\\
  \end{array}
\right),
\end{equation}
where $\Gamma_A$, $\Gamma_{\bar{A}}$ describe the connections inside each subsystem, and the off-diagonal submatrix $\Gamma_{A\bar{A}}$ is for the ones between them.

Given a graph state $\ket{G}$ with its associated graph $G$, the reduced density matrix of a subsystem $A$ is $\rho_A=\mathrm{Tr}_{\bar{A}}(\ket{G}\bra{G})$, where the partial trace is on $\bar{A}$. The explicit formula of the entanglement entropy is
\begin{equation}\label{Eq:Entfor}
 S(\rho_A)=\mathrm{rank}(\Gamma_{A\bar{A}})
\end{equation}
where the $\mathrm{rank}$ is on the binary field $\mathbb{F}_2$, and $S(\rho)=-\mathrm{Tr}[\rho\log_2\rho]$ is Von Neumann entropy. Note that Renyi-$\alpha$ entropy $S_\alpha(\rho_A)$ of any order is also suitable here, since the spectrum of $\rho_A$ is flat for graph states \cite{Hein2004Multiparty}.

\subsection*{Examples: 1-D and 2-D cluster states}
Now we apply Theorem 1 to detect entanglement structures of 1-D and 2-D cluster states. The corresponding graphs of these states are all $2$-colorable, i.e., $k=2$, thus one can realize the witnesses with only two local measurement settings. For clearness, the vertexes in the subsets $V_1$ and $V_2$ are associated with red and blue colors respectively, as shown in Fig.~\ref{Fig:allstate}. According to Eq.~(10) in Main text, we write the projectors for each subset as,
\begin{equation}\label{Eq:P1P2}
\begin{aligned}
P_1 &= \prod_{red\ i}  \frac{S_i+\mathbb{I}}{2},\\
P_2 &= \prod_{blue\ i} \frac{S_i+\mathbb{I}}{2}.
\end{aligned}
\end{equation}
%\subsubsection{1-D cluster state}

We first detect the entanglement structures of the 1-D cluster state $\ket{C_1}$ using the two projectors $P_1$ and $P_2$ defined in Eq.~\eqref{Eq:P1P2} with stabilizers of $\ket{C_1}$. Consider an example of tripartition $\mathcal{P}_3=\{A_1,A_2,A_3\}$ as shown in Fig.~\ref{Fig:allstate}(a), there are three ways to divide the subsystems into two sets, i.e., $\{A,\bar{A}\}=\{A_1, A_2A_3\}$, $\{A_2, A_1A_3\}$, $\{A_3, A_1A_2\}$. And the corresponding entanglement entropies are $S(\rho_{A_1})=S(\rho_{A_3})=1$ and $S(\rho_{A_2})=2$ , which is a manifest of the area law of entanglement entropy \cite{Eisert2010area}. Thus the maximal and minimal entropy is $2$ and $1$.
According to Theorem 1, the operators to witness $\mathcal{P}_3$-entanglement structure are given by,
\begin{equation}\label{Eq:witen1D}
\begin{aligned}
W_{f,C_1}^{\mathcal{P}_3} &= \frac{5}{4}\mathbb{I}-(P_1+P_2), \\
W_{b,C_1}^{\mathcal{P}_3} &=\frac{3}{2}\mathbb{I}-(P_1+P_2),
\end{aligned}
\end{equation}
where the two projectors $P_1$ and $P_2$ are defined in Eq.~\eqref{Eq:P1P2} with the graph of Fig.~\ref{Fig:allstate}(a). Similar analysis works for other general partitions.

Here we show how to calculate the entanglement entropy $S(\rho_{A_2})$ by using the formula in Eq.~\eqref{Eq:Entfor}, and the entanglement entropy of other subsystems can be calculated similarly.
The matrix $\Gamma_{A_2,\bar{A_2}}$ which describes the connections between $A_2$ and $\bar{A_2}=A_1A_3$ shows,
%\begin{equation}
%\Gamma_{A_2,\bar{A_2}}=
%\begin{array}{lc}
%\mbox{}&
%\begin{array}{cccc}1&4&5&6 \end{array}\\
%\begin{array}{c}2\\3\end{array}&
%\left(\begin{array}{cccc}
%1 & 0 & 0 &0 \\
%0& 1 & 0 & 0\\
%\end{array}\right)
%\end{array}
%\end{equation}
\begin{equation}
\Gamma_{A_2,\bar{A_2}}=
\bordermatrix{%
&1&4&5&6\cr
2&1 & 0 & 0 &0\cr
3&0& 1 & 0 & 0\cr
},
\end{equation}
where the indexes of the row and column label the vertexes in $A_2$ and $\bar{A_2}$ respectively. In fact, one just needs to consider following submatrix,
\begin{equation}\label{Eq:1DsubM}
\bordermatrix{%
&1&4\cr
2&1& 0\cr
3&0& 1\cr
},
\end{equation}
which is obtained from $\Gamma_{A_2,\bar{A_2}}$ by deleting the last two columns, and shares the same rank with $\Gamma_{A_2,\bar{A_2}}$.  Note that there is no edge between the vertexes $5,6$ in $\bar{A_2}$ and the ones in $A_2$. The matrix in Eq.~\eqref{Eq:1DsubM} is an identity matrix with rank $2$, thus $S(\rho_{A_2})=2$.
Generally speaking, one only needs to consider the vertexes which are related to the edges crossing $A$ and $\bar{A}$, when calculating the rank of the matrix $\Gamma_{A,\bar{A}}$.

Next, we consider the 2-D cluster state $\ket{C_2}$ that is defined on a $2$-D square lattice. As an example, we consider a tripartition of a $5\times5$ lattice as shown in Fig.~\ref{Fig:allstate}(b).
Similar as the 1-D cluster state case, the corresponding entanglement entropies are $S(\rho_{A_1})=S(\rho_{A_3})=5$ and $S(\rho_{A_2})=4$. According to Theorem 1, the operators to witness $\mathcal{P}_3$-entanglement structure are given by,
\begin{equation}\label{Eq:witen2D}
\begin{aligned}
W_{f,C_2}^{\mathcal{P}_3} &= \frac{33}{32}\mathbb{I}-(P_1+P_2), \\
W_{b,C_2}^{\mathcal{P}_3} &= \frac{17}{16}\mathbb{I}-(P_1+P_2),
\end{aligned}
\end{equation}
where the two projectors $P_1$ and $P_2$ are defined in Eq.~\eqref{Eq:P1P2} with the graph of Fig.~\ref{Fig:allstate}(b). Similar analysis works for other general partitions.

Here we show how to calculate the entanglement entropy $S(\rho_{A_2})$ by using the formula in Eq.~\eqref{Eq:Entfor}, and the entanglement entropy of other subsystems can be calculated similarly. As mentioned in the 1-D case, one only needs to consider the following matrix which shares same rank with $\Gamma_{A_2,\bar{A_2}}$
%\begin{equation}
%\Gamma_{A_2,\bar{A_2}}=
%\begin{array}{lc}
%\mbox{}&
%\begin{array}{cccc}1&4&5&6 \end{array}\\
%\begin{array}{c}2\\3\end{array}&
%\left(\begin{array}{cccc}
%1 & 0 & 0 &0 \\
%0& 1 & 0 & 0\\
%\end{array}\right)
%\end{array}
%\end{equation}
\begin{equation}\label{Eq:2DsubM}
\bordermatrix{%
&5&6&7&8&9\cr
1&1 & 0 & 0 &0&0\cr
2&0& 1 & 0 & 0&0\cr
3&0 & 0 & 1 &1&0\cr
4&0& 0 & 0 & 0&1\cr
},
\end{equation}
which is obtained from $\Gamma_{A_2,\bar{A_2}}$ by ignoring the vertexes in the bulks of both $A_2$ and $\bar{A_2}$. It is clear that the rank of the matrix in Eq.~\eqref{Eq:2DsubM} is $4$, since the four row vectors are linearly independent. As a result, $S(\rho_{A_2})=4$.

%\begin{figure}[t]
%\centering
%\includegraphics[width=0.45\textwidth]{allstateApp-eps-converted-to.pdf}
%\caption{Graphs of (a) 1-D cluster state,  (b) 2-D cluster state, and (c) GHZ state}\label{Fig:allstate}
%\end{figure}

\subsection{Witnessing subsystem entanglement structures}
\begin{figure}[t]
\centering
\includegraphics[width=0.45\textwidth]{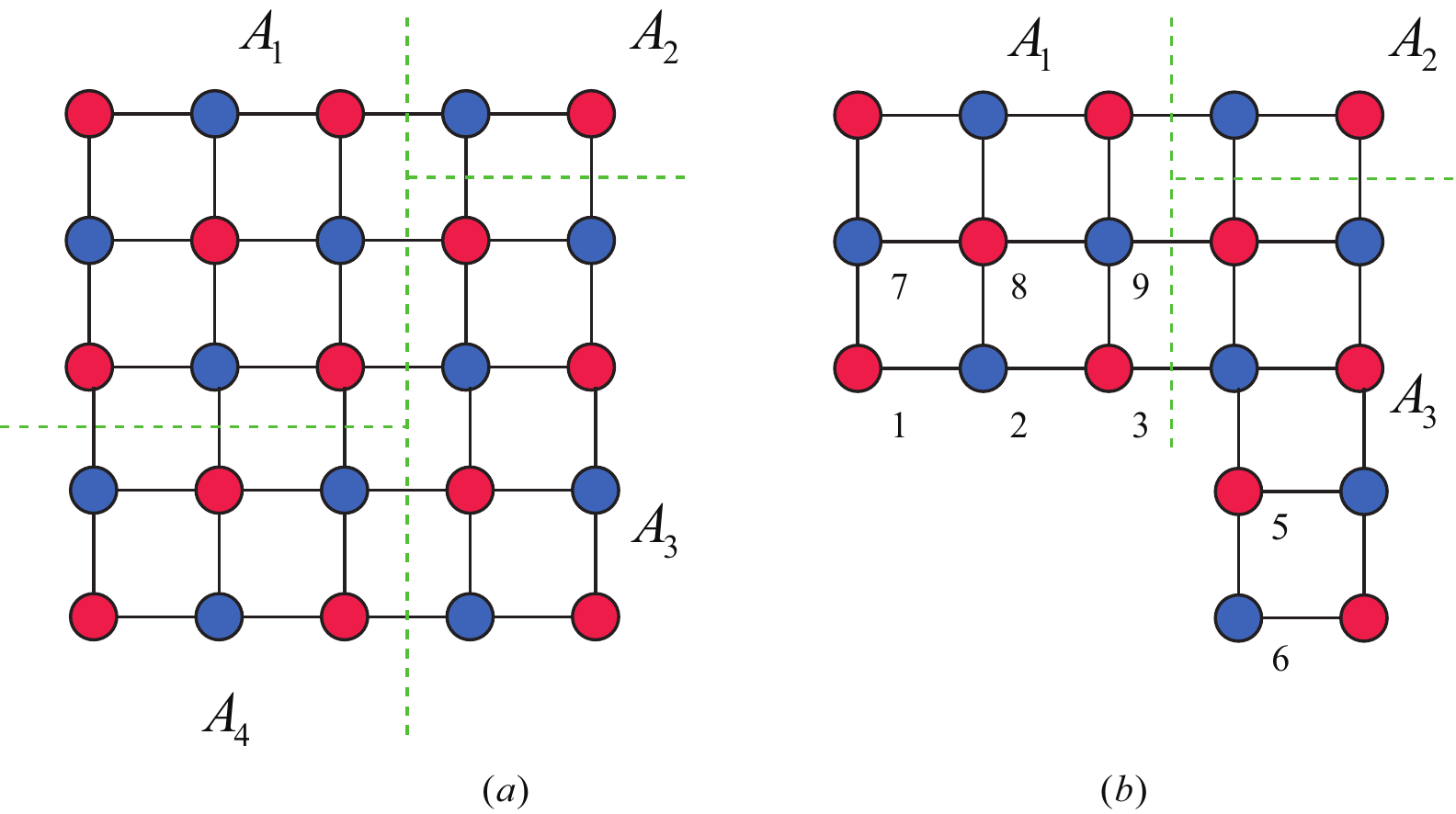}
\caption{Witnessing subsystem entanglement structures.}\label{Fig:subsystem}
\end{figure}
As mentioned in Main Text, the witnesses in Theorem 1 can also be used to detect the entanglement structures of a few of subsystems, by only post-processing the original measurement results. In the following, we take the 2-D cluster state to illustrate this advantage.

As shown in Fig.~\ref{Fig:subsystem}(a), the whole 2-D lattice is partitioned into four parts  $\{A_1,A_2,A_3, A_4\}$. Suppose one only cares about the entanglement structures among subsystems $A_1, A_2$ and $A_3$, the witnesses shown in Theorem 1 can also be applied in this scenario.
%To be specific, let us denote $\rho'=\mathrm{Tr}_{A_2}(\rho)$, which is the reduced density matrix from the shared state $\rho$ on subsystems $A_1, A_3$ and $A_4$.
To be specific, the projectors $P_1$ and $P_2$ appearing in the witnesses should be changed to $P_1'$ and $P_2'$, which are related to the new graph state $\ket{G_{A_1A_2A_3}}$. $G_{A_1A_2A_3}$ is the subgraph of $G$ with edges and vertexes related to $A_2$ deleted, i.e., $G_{A_1A_2A_3}=G-V_{A_4}$, as shown in Fig.~\ref{Fig:subsystem} (b). Similar as the 2-D cluster case shown previously, now we should calculate the entanglement entropy with respect to the new graph state $\ket{G_{A_1A_2A_3}}$. It is not hard to find that $S(\rho'_{A_1})=S(\rho'_{A_3})=3$ and $S(\rho'_{A_2})=2$, where $\rho'_{A_i}=\mathrm{Tr}_{\bar{A_i}}(\ket{G_{A_1A_2A_3}}\bra{G_{A_1A_2A_3}})$ for $i=1,2,3$.

As a result, according to Theorem 1, the operators to witness this $\mathcal{P}_3=\{A_1,A_2,A_3\}$-entanglement structure are given by,
\begin{equation}\label{Eq:witsub}
\begin{aligned}
W_{f,sub}^{\mathcal{P}_3} &= \frac{9}{8}\mathbb{I}-(P_1'+P_2'), \\
W_{b,sub}^{\mathcal{P}_3} &= \frac{5}{4}\mathbb{I}-(P_1'+P_2'),
\end{aligned}
\end{equation}
where the two projectors $P_1'$ and $P_2'$ are defined in Eq.~\eqref{Eq:P1P2} with stabilizers of the graph state $\ket{G_{A_1A_2A_3}}$. Now the stabilizers which constitute the projectors are restricted on the subsystem $A_1A_2A_3$. For example, $S_1'=X_1Z_2Z_7$ and $S_2'=X_2Z_1Z_8Z_3$ in $P_1'$ and $P_2'$ in Fig.~\ref{Fig:subsystem} (b).

Note that the expectation values of them can also be evaluated from the two original local measurement settings $\bigotimes_{i\in V_1} X_i \bigotimes_{j\in V_2} Z_j$ and $\bigotimes_{i\in V_1} Z_i \bigotimes_{j\in V_2} X_j$, which are employed to measure $P_1$ and $P_2$, respectively. Consequently, one can detect the entanglement structures of any subset of subsystems by only post-processing the measurement results, without conducting the experiment again.
\subsection*{Multi-color graph state}
As shown in Main Text, the number of local measurement settings in the detection is directly related to the colorability of the corresponding graph. That is, the witnesses can be realized with $k$ local measurements when the corresponding graph $G$ is $k$-colorable. We have shown several widely-used graph states whose graphes are 2-colorable. Here we give a $3$-colorable graph state and construct the entanglement structure witnesses according to Theorem 1.

Before that, we remark that one may reduce the chromatic number of the underlying graph by applying local complementation, which can be realized by local Clifford operation on the graph state \cite{Van2004Graphical,Hein2006Graph}. To be specific, local complementation $\tau_i$ on $G$ with respective to the vertex $i$ is to delete edges between vertexes in the neighborhood set $N_i$ if they are originally connected; or to add edges otherwise. The corresponding local Clifford unitary to realize this graph transformation shows,
\begin{equation}\label{}
U_i(G)=\exp(-i\frac{\pi}{4}X_i)\bigotimes_{j\in N_i}\exp(i\frac{\pi}{4}Z_j),
\end{equation}
and one has $\ket{\tau_i(G)}=U_i(G)\ket{G}$.
For example, as shown in Fig.~\ref{Fig:kcolor} (a), the fully connected graph can be transformed to the star graph by local complementation, and one is $N$-colorable and the other is $2$-colorable.

Then let us consider the five-qubit ring state $\ket{R_5}$ shown in Fig.~\ref{Fig:kcolor}(b), with the corresponding graph being 3-colorable. The chromatic number decides how many local measurements that one needs to lower bound the fidelity between the separable states with the graph state, as shown in Proposition 2. In this case, we needs 3 local measurement settings according to the red, blue, and yellow vertexes to obtain the expectation values of the following three projectors,
\begin{equation}\label{Eq:P1P2P3}
\begin{aligned}
P_1 &= \prod_{red\ i}  \frac{S_i+\mathbb{I}}{2},\\
P_2 &= \prod_{blue\ i} \frac{S_i+\mathbb{I}}{2},\\
P_3 &= \prod_{yellow\ i} \frac{S_i+\mathbb{I}}{2}.
\end{aligned}
\end{equation}

On the other hand, the optimizations in the witnesses shown in Theorem 1 just relate to the entanglement entropy of the graph state, not the chromatic number of it. The entanglement entropies of the subsystems show $S(\rho_{A_1})=1$ and $S(\rho_{A_2})=S(\rho_{A_3})=2$. According to Theorem 1, the operators to witness  $\mathcal{P}_3$-entanglement structure are given by,
\begin{equation}\label{Eq:wit}
\begin{aligned}
W_{f,R_5}^{\mathcal{P}_3} &= \frac{9}{4}\mathbb{I}-(P_1+P_2+P_3),\\
W_{b,R_5}^{\mathcal{P}_3} &= \frac{5}{2}\mathbb{I}-(P_1+P_2+P_3),
\end{aligned}
\end{equation}
where the two projectors $P_1, P_2$ and $P_3$ are defined in Eq.~\eqref{Eq:P1P2P3} with stabilizers of the ring state $\ket{R_5}$.
\begin{figure}[t]
\centering
\includegraphics[width=0.45\textwidth]{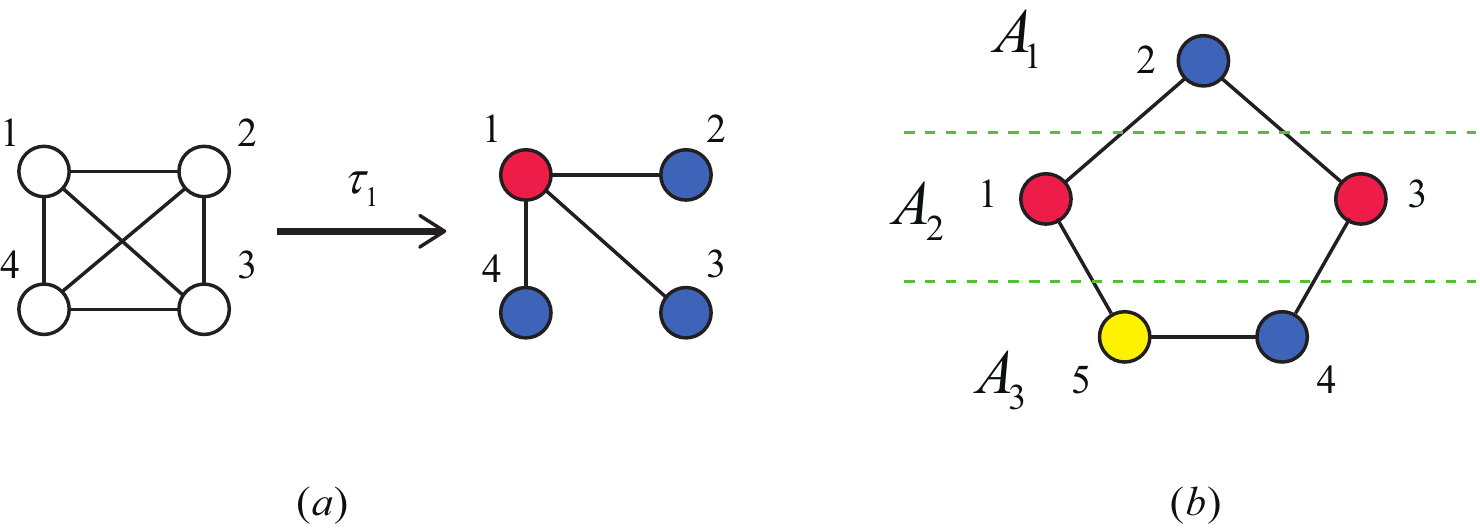}
\caption{(a) The fully connected graph which is $N$-colorable, and the local complementation can transform it to the star graph which is 2-colorable. (b) Five-vertex ring graph is $3$-colorable, and the whole system is partitioned into 3 parts.}\label{Fig:kcolor}
\end{figure}

\section*{SUPPLEMENTARY Methods 1: Proof of Corollary \ref{Th:GHZ}}\label{Sec:Co:one}
\begin{proof}
The genuine entanglement witness $W_{b}^{\mathcal{P}_N}$ in Eq.~\eqref{Eq:GHZgenuine} is obtained from the second equation in Theorem 1, by considering the $N$-partition $\mathcal{P}_N$. Thus, we need to find the solution of $\max_{\{A,\bar{A}\}} 2^{-S(\rho_A)}$, or equivalently, $\min_{\{A,\bar{A}\}} S(\rho_A)$ of the $N$-partite graph state $\ket{G}$.
Here we show that $\min_{\{A,\bar{A}\}} S(\rho_A)=1$. Since $\ket{G}$ is an entangled pure state, $S(\rho_A)>0$ for any $A$. Because the entanglement spectrum is flat for any graph state, $S(\rho_A)$ is at least $1$ with spectrum $\{\frac{1}{2}, \frac{1}{2}\}$. We can choose any single qubit as $A$ such that $S(\rho_A)=1$ and finish the proof.
\end{proof}

In the following, we show that the witness $W_{b}^{\mathcal{P}_N}$ is tight (or optimal), in the sense that there exists a bi-separable state $\rho_b$ that saturates $\mathrm{Tr}(\rho_b W_{b}^{\mathcal{P}_N})=0$. For simplicity, we consider the case where the underlying graph is $k=2$-colorable. For general $k$ case, the tightness can be proved similarly. As $k=2$, the witness in Eq.~\eqref{Eq:GHZgenuine} becomes,
\begin{equation}\label{Eq:GHZgenuine2}
W_{b}^{\mathcal{P}_N} = \frac{3}{2}\mathbb{I}-(P_1+P_2),
\end{equation}
which is suitable for 1-D and 2-D cluster states and the GHZ state shown in Fig.~\ref{Fig:allstate}, and the projectors $P_1$ and $P_2$ are defined in Eq.~\eqref{Eq:P1P2}.

To show the tightness, we give a specific bi-separable state such that $\mathrm{Tr}(\rho_b W_{b}^{\mathcal{P}_N})=0$, or equivalently $\langle P_1+P_2\rangle =\frac{3}{2}$.
$P_1$ and $P_2$ can be written explicitly as a summation of stabilizers in the form $K_{\vec{p}}=S_1^{p_1}S_2^{p_2}\cdots S_N^{p_N}$, with $\vec{p}=(p_1,p_2\cdots p_N)$ a binary vector. To be specific, $P_1$ and $P_2$ contains all the stabilizers generated by the independent stabilizers $S_i$ of the red and blue vertexes respectively,
\begin{equation}\label{Eq:P1P2new}
\begin{aligned}
&P_1=\sum_{red\ \vec{p}} \frac{K_{\vec{p}}}{2^{n_r}},\\
&P_2=\sum_{blue\ \vec{q}} \frac{K_{\vec{q}}}{2^{n_b}},
\end{aligned}
\end{equation}
where $n_r$ and $n_b$ denote the number of red and blue vertexes respectively, with $n_r+n_b=N$;
$\vec{p}$ denote red type vector whose $p_i=0$ for all blue $i$; $\vec{q}$ denote blue type vector whose $q_{i'}=0$ for all red $i'$.

Now suppose the first vertex is red, we choose $\ket{\Psi_{b}}=\ket{0}_1\otimes \ket{G_{\{2,3,\cdots,N\}}}$, where the first qubit is set as $\ket{0}$ and the remaining qubits hold a graph state $\ket{G_{\{2,3,\cdots,N\}}}$. The corresponding graph $G_{\{2,3,\cdots,N\}}$ is obtained from the original graph $G$ via deleting the vertex 1 and the edges connected to it. First,
considering the stabilizer $K_{\vec{p}}$ in the projector $P_1$. If $K_{\vec{p}}$ does not contain
$S_1$, one has $\langle K_{\vec{p}} \rangle=1$, since it is still a stabilizer of the state $\ket{G_{\{2,3,\cdots,N\}}}$, when restricted on the qubits $\{2,\cdots,N\}$; otherwise $\langle K_{\vec{p}} \rangle=0$, since $S_1$ contains a Pauli $X_1$ on the first qubit and $\bra{0}X\ket{0}=0$.
As a result, we have $\langle P_1\rangle=\frac1{2}$, because there are one half of $K_{\vec{p}}$ containing $S_1$. Then, considering the stabilizer $K_{\vec{q}}$ in the projector $P_2$. $\langle K_{\vec{q}}\rangle=1$ for any $K_{\vec{q}}$ in $P_2$, since $K_{\vec{q}}$ contains a Pauli $Z_1$ or $\mathbb{I}_1$ on the first qubit with $\bra{0}Z(\mathbb{I})\ket{0}=1$, and the remaining part of it is actually a stabilizer of the state $\ket{G_{\{2,3,\cdots,N\}}}$. As a result, one has
$\langle P_2\rangle=1$ and hence $\langle P_1+P_2\rangle=\frac{3}{2}$.

\section*{SUPPLEMENTARY Methods 2: Proof of Corollary \ref{Th:1DN}}\label{Sec:Co:two}
As mentioned in Main Text, Corollary \ref{Th:1DN} is obtained from the first equation in Theorem 1 by taking all the subsystems just containing one qubit, i.e., an $N$-partition $\mathcal{P}_N$.
To prove it, one needs to find the solution of the optimization problem about entanglement entropy, say, $\min_{\{A,\bar{A}\}} 2^{-S(\rho_A)}$, where $A$ is a subsystem and $S(\rho_A)=\mathrm{Tr}_{\bar{A}}(\ket{C}\bra{C})$. Equivalently, one should find the solution of the optimization $\max_{\{A,\bar{A}\}}S(\rho_A)$, denoted as $S_{max}$ for simplicity. In the following, we prove Corollary \ref{Th:1DN} for 1-D and 2-D cluster states, by showing that $S_{max}=\lfloor\frac{N}{2}\rfloor$, respectively.
\subsection*{Proof of Corollary \ref{Th:1DN} of 1-D cluster state}\label{subsec:Cor1}
\begin{proof}
Here, we show that $S_{max}=\lfloor\frac{N}{2}\rfloor$ for the 1-D cluster state. First, note the entanglement entropy should be no more than the qubit number in the subsystem, thus one has $S(\rho_A)\leq |A|$ and $S(\rho_{\bar{A}})\leq|\bar{A}|$. On account of $S(\rho_A)=S(\rho_{\bar{A}})$, one further has $S(\rho_A)\leq \min\{|A|, |\bar{A}|\}$. As a result, $S_{max}\leq\lfloor\frac{N}{2}\rfloor$.

Then, we choose the system $A$ composed of all the qubits on the odd sites to saturate this bound. One can calculate the corresponding entanglement entropy with the formula given in Eq.~\eqref{Eq:Entfor}. Here, we evaluate $S(\rho_A)$ with another more intuitive method by distilling EPR pairs between $A$ and $\bar{A}$ with local unitary operations \cite{Briegel2001Persistent,Damian2007Entanglement}. Since local unitary does not change the entanglement entropy, one can properly evaluate it by counting the final number of EPR pairs.

As shown in Fig.~\ref{Fig:1Ddistill}, we divide the system into two parts according to the odd/even position of the qubits (or red/blue according to Fig.~\ref{Fig:allstate}). First, apply ``local'' unitary Controlled-Z operation $CZ^{\{1,3\}}$ on qubits $1$ and $3$, where locality is in the sense that one considers $A$ and $\bar{A}$ as two subsystems. Second, apply local complementation on qubit $1$, which can be realized by local Clifford unitary \cite{Van2004Graphical,Hein2006Graph}. Local complementation $\tau_i$ on $G$ with respective to the vertex $i$ is to add edges between the vertexes in the neighborhood set $N_i$ under module 2. As a result, the edge $\{2,3\}$ is deleted. Third, apply $CZ^{\{1,3\}}$ again and there is a EPR pair appearing between qubits $1$, $2$.  Here we call the two-qubit graph state $\frac{1}{\sqrt{2}}(\ket{0+}+\ket{1-})$ EPR pair without confusion. Iterating this process, one finally can distill $\lfloor\frac{N}{2}\rfloor$ EPR pairs.
Consequently, one has $S(\rho_{odd})=\lfloor\frac{N}{2}\rfloor$.
\end{proof}
\begin{figure}[thb]
\centering
\resizebox{8cm}{!}{\includegraphics[scale=1]{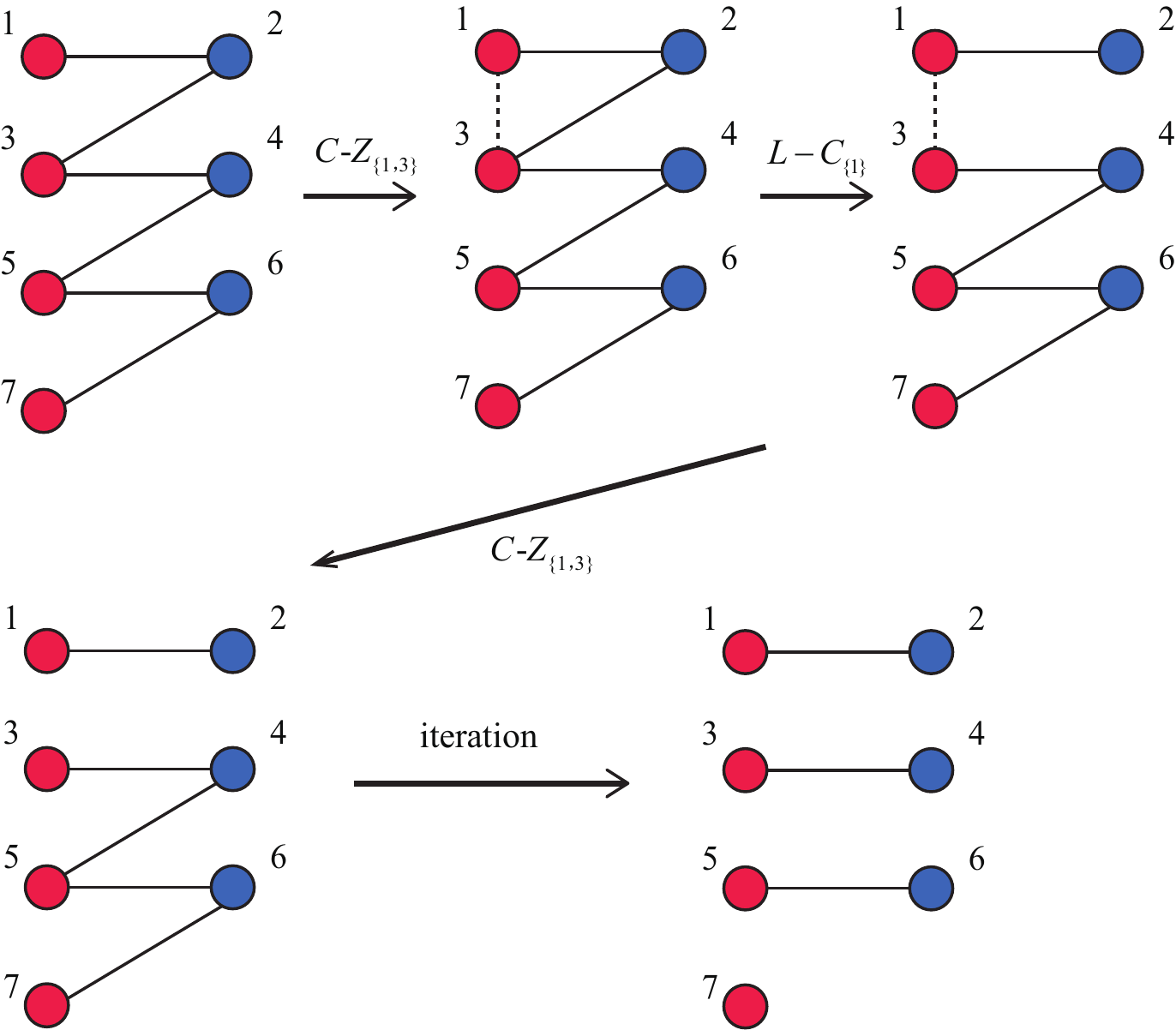}}
\caption{An illustration of distilling EPR pairs from the 1-D cluster state by local unitary operations.}
\label{Fig:1Ddistill}
\end{figure}
\subsection*{Proof of Corollary \ref{Th:1DN} of 2-D cluster state}\label{subsec:Cor12D}
\begin{proof}
Here, we show that $S_{max}=\lfloor\frac{N}{2}\rfloor$ for 2-D cluster state.  One has $S_{max}\leq\lfloor\frac{N}{2}\rfloor$ as in the 1-D case, since the entanglement entropy can not be larger than the number of qubit in the subsystems. In the following, we give examples of subsystem $A$ to saturate this bound. The chosen $A$ depends on the even/odd property of the total number of qubits.

First, let us consider the even-qubit case, i.e., the qubit number $N=n\times n$ where $n$ is even. We select all the odd columns to constitute the subsystem $A$, as shown in Fig.~\ref{Fig:even2D}. Then one can prove that $S(\rho_A)=\frac{N}{2}$ by distilling EPR pairs as in the 1-D case. To be specific, one first applies $CZ$ operations inside both subsystems $A$ and $\bar{A}$, such that the quantum state becomes $n$ independent 1-D cluster states each with $n$ qubits. Then by distilling EPR for each 1-D cluster state, finally one can obtain totally $n\times \frac{n}{2}=\frac{N}{2}$ EPR pairs, thus $S(\rho_A)=\frac{N}{2}$.

For the odd-qubit case, $N=n\times n$ with $n$ being odd. The previous choosing style can not make $S(\rho_A)=\lfloor\frac{N}{2}\rfloor$. Here we show a modified one. The first $n-1$ columns belong to $A$ and $\bar{A}$ in succession as before. For the qubits in the final column, we successively distribute the qubits to $A$ and $\bar{A}$. Then as in the even-qubit case, one applies $CZ$ operations inside both subsystems and distill EPR pairs for several 1-D cluster states, as shown in Fig.~\ref{Fig:even2D}. Finally, we also get $S(\rho_A)=\lfloor\frac{N}{2}\rfloor$.
\end{proof}
We remark that Corollary \ref{Th:1DN} also holds for the 2-D cluster state with general rectangle lattice $n_1\times n_2$. With loss of generality, assume that $n_1$ is odd and $n_2$ is even, one can distill $\lfloor\frac{N}{2}\rfloor$ EPR pairs by choosing $A$ containing all the odd columns as in the above proof. We further conjecture that Corollary \ref{Th:1DN} holds for any (such as 3-D) cluster states.

Finally, we show that the witness $W_{f,C}^{\mathcal{P}_N}$ in Corollary \ref{Th:1DN} is tight.
We give a fully separable state $\ket{\Psi_{f}}=\bigotimes_{red\ i} \ket{+}_i \bigotimes_{blue\ j} \ket{0}_j$ to saturate the bound, where the red (blue) qubits are set as $\ket{+}$ ($\ket{0}$) respectively. Similar as the discussion of $W_{b}^{\mathcal{P}_N}$ in Sec.~\ref{Sec:Co:one}, one has $\langle P_1\rangle=1$, since $\langle K_{\vec{p}} \rangle=1$ for all $K_{\vec{p}}$ in $P_1$; $\langle P_2\rangle=2^{-\lfloor\frac{N}{2}\rfloor}$, since all $\langle K_{\vec{q}} \rangle=0$ in $P_2$ except the $\mathbb{I}$. Consequently, $\langle P_1+P_2\rangle=1+2^{-\lfloor\frac{N}{2}\rfloor}$.
\begin{figure}[thb]
\centering
\resizebox{7cm}{!}{\includegraphics[scale=1]{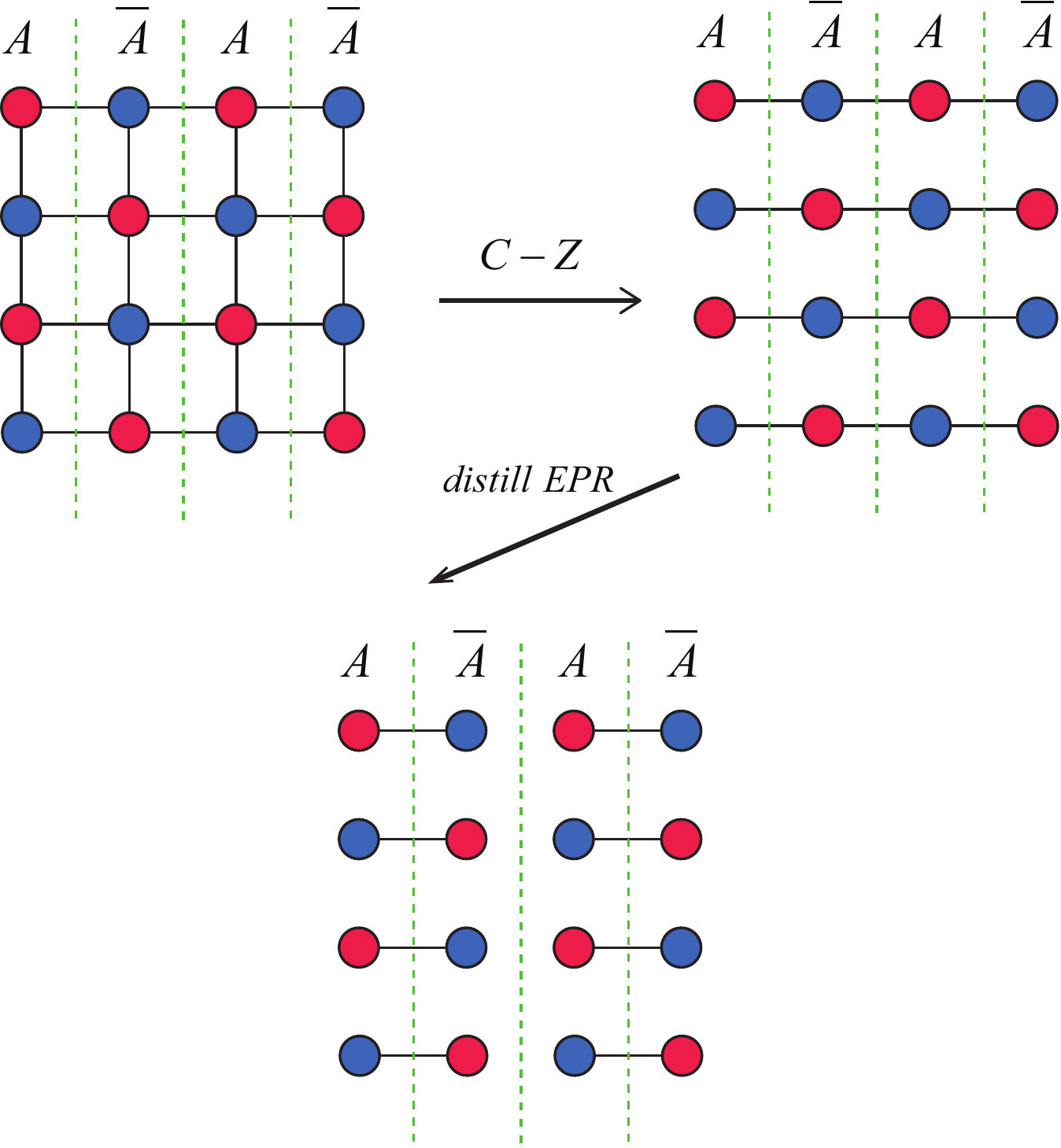}}
\caption{Distillation of EPR pairs in the $4\times 4$ 2-D cluster state}
\label{Fig:even2D}
\end{figure}
\begin{figure}[thb]
\centering
\resizebox{9cm}{!}{\includegraphics[scale=1]{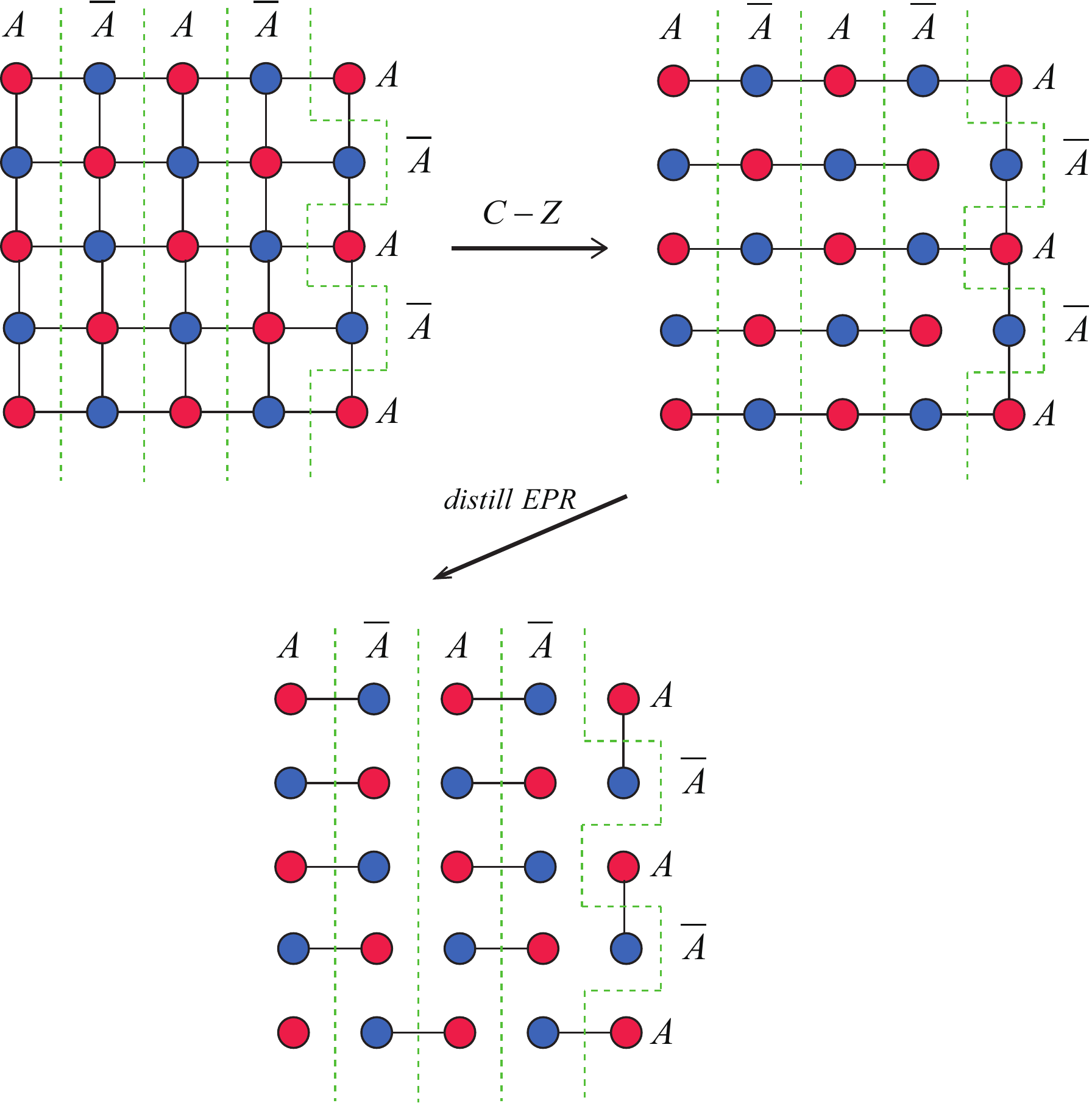}}
\caption{Distillation of EPR pairs in the $5\times 5$ 2-D cluster state}
\label{Fig:odd2D}
\end{figure}
\section*{SUPPLEMENTARY Methods 3: Proof of Corollary \ref{Th:1Dmsep}}\label{Sec:Co:three}
Before proving Corollary \ref{Th:1Dmsep}, we introduce the following Lemma. It gives a lower bound of $S(\rho_A)$ related to the geometric connection between $A$ and $\bar{A}$. First, let us introduce the definition of the boundary between $A$ and $\bar{A}$ on a 1-D qubit chain. The edge $(i,i+1)$ is called a boundary, if $i\in A$, $i+1\in \bar{A}$ or vice versa.

\begin{lemma}\label{Lem:Entropy}
Given a bipartition $\{A,\bar{A}\}$  on a 1-D qubit chain, if the number of boundaries between $A$ and $\bar{A}$ is not less than $m-1$, the entanglement entropy can be lower bounded by,
\begin{equation}\label{Eq:LB:Sa}
S(\rho_A)\geq \lfloor\frac{m}{2}\rfloor,
\end{equation}
where $\rho_A=\mathrm{Tr}_{\bar{A}}(\ket{C_1}\bra{C_1})$ is the the reduced density matrix of a 1-D cluster state $\ket{C_1}$.
\end{lemma}
Note that Lemma \ref{Lem:Entropy} can be seen as a manifest of the area law of entanglement entropy \cite{Eisert2010area}.
\begin{figure}[thb]
\centering
\resizebox{7cm}{!}{\includegraphics[scale=1]{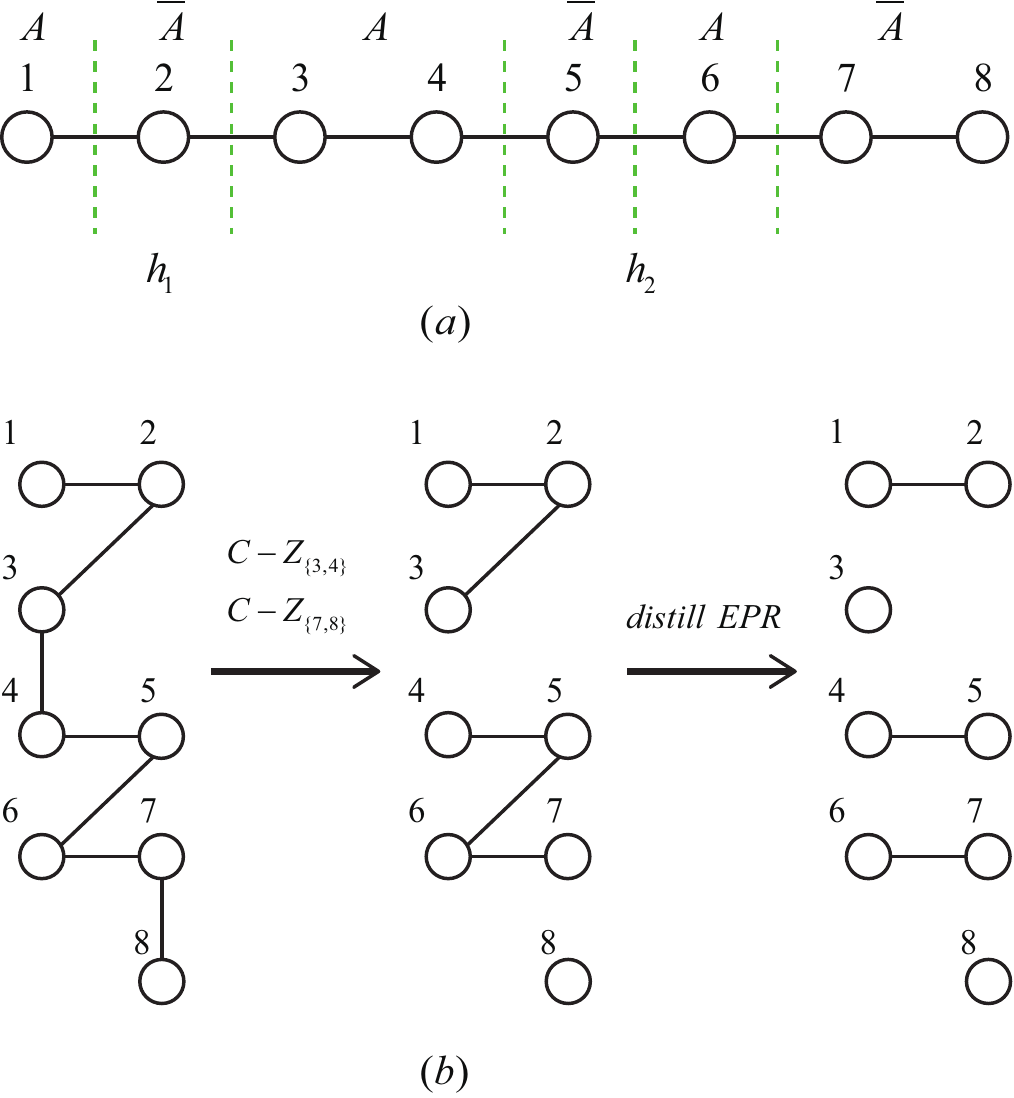}}
\caption{(a) Given a bipartition $\{A,\bar{A}\}$, subsystems $A$ and $\bar{A}$ distribute sequentially on the chain. In this 7-qubit example, there are two boundary clusters $h_1$ and $h_2$, with $|h_1|=2$ and $|h_2|=3$. (2) Distillation of EPR pairs between $A$ and $\bar{A}$. First, one can apply C-Z operation inside $A$ and $\bar{A}$ to eliminate the irrelevant edges. Here we apply $C-Z_{\{3,4\}}$ and $C-Z_{\{7,8\}}$. Second, we distill EPR pairs between $A$ and $\bar{A}$ as in Fig.~\ref{Fig:1Ddistill}, and one can see that the distillation is conducted inside each cluster.}
\label{Fig:1Dmsep}
\end{figure}
\begin{proof}
In the following, we bound the value of $S(\rho_A)$ by distilling EPR pairs between $A$ and $\bar{A}$ with local unitary operations, similar as the proof of Corollary \ref{Th:1DN}.

Given any bipartition $\{A,\bar{A}\}$, subsystems $A$ and $\bar{A}$ would distribute sequentially on the chain, as shown in Fig.~\ref{Fig:1Dmsep}. Two boundaries may be close to each other, such as $(i,i+1)$ and  $(i+1,i+2)$. There are several boundary clusters, denoted by $h_k$. Inside each cluster, the boundaries distribute in sequence closely; between the clusters, the boundaries are far away from each other more than one qubit. Fig.~\ref{Fig:1Dmsep} (a) shows a 7-qubit example, where there are two boundary clusters $h_1$ and $h_2$.

Via applying local unitary operations, one can distill $\lceil |h_k|/2 \rceil$ EPR pairs from each cluster, where $|h_k|$ is the number of boundaries (see Fig.~\ref{Fig:1Dmsep} (b) for a detailed illustration). As a result, the entanglement entropy is bounded by,
\begin{equation}\label{}
\begin{aligned}
S(\rho_A)=\sum_k \left\lceil \frac{|h_k|}{2}\right\rceil\geq \left\lceil\sum_k \frac{|h_k|}{2}\right\rceil\geq \left\lceil \frac{m-1}{2}\right\rceil=\left\lfloor\frac{m}{2}\right\rfloor£¬
\end{aligned}
\end{equation}
where the second inequality is because the precondition that the number of boundaries is not less than $m-1$.
\end{proof}
Then we give the proof of Corollary \ref{Th:1Dmsep} as follows.
\begin{proof}
According to Theorem 2, one needs to prove the solution of the optimization,
\begin{equation}\label{}
f_1(m)\equiv \min_{\mathcal{P}_m}\max_{\{A,\bar{A}\}} S(\rho_A)=\left\lfloor\frac{m}{2}\right\rfloor
\end{equation}
where the maximization takes over all possible partitions $\mathcal{P}_m$ with $m$ subsystems, the minimization takes over all bipartitions of $\mathcal{P}_m$, and  $\rho_A=\mathrm{Tr}_{\bar{A}}(\ket{C_1}\bra{C_1})$. In the following, We first show $f_1(m)\geq\left\lfloor\frac{m}{2}\right\rfloor$ and then give a specific partition to saturate the bound.

Given any m-partition $\mathcal{P}_m=\{A_i\}_{i=1,2\cdots m}$, one can always choose a bipartition $\{A,\bar{A}\}$ of $\mathcal{P}_m$,
such that $A$ and $\bar{A}$ shares at least $m-1$ boundaries. Then, with Lemma \ref{Lem:Entropy} one has $f_1(m)\geq\left\lfloor\frac{m}{2}\right\rfloor$. To be specific, if $A_i$ distribute sequentially on the qubit chain, one can choose $A_i$ with odd $i$ to constitute $A$, i.e., $A=\bigcup_{odd\ i} A_i$, and it is clear that $A$ and $\bar{A}$ shares $m-1$ boundaries in this case. For general cases, by utilizing a fact of graph theory, one can also find a proper bipartition of $\mathcal{P}_m$ as follows. Let every subsystem $A_i$ represent a vertex $i$ of a weighted graph, and the weight of the edge $(i,j)$ is the number of boundaries between $A_i$ and $A_j$. Since this graph is connected, one can choose a (minimum) spanning tree of this graph \cite{Kleinberg2005Algorithm}. Then let $A$ contains all the $A_i$ in the odd layers of the tree, and there are $m-1$ edges between $A$ and $\bar{A}$. As a result, the number of boundary is at least $m-1$, since there is at least one boundary with respect to one edge. %\red{maybe draw a fig to illustrate}.

Finally, we give a specific partition to show that $f_1(m)=\left\lfloor\frac{m}{2}\right\rfloor$. Let $A_1=\{1\}, A_2=\{2\},\cdots, A_m=\{m,m+1,\cdots,N\}$, where $1,2,\cdots N$ denote the qubits. Then choose $A=\bigcup_{odd\ i} A_i$, one can distill EPR pairs as in the proof of Lemma \ref{Lem:Entropy}. Now there is only one boundary cluster and one has $S(\rho_A)=\left\lfloor\frac{m}{2}\right\rfloor$. Finally, we finish the proof of $f_1(m)=\left\lfloor\frac{m}{2}\right\rfloor$ as well as Corollary \ref{Th:1Dmsep}.
\end{proof}

At the end of this section, we show that the witnesses in Corollary \ref{Th:1Dmsep} are tight, in the sense that there is an $m$-separable state $\rho_m$ to saturate $\mathrm{Tr}(\rho_mW_{m,C_1})=0$. Here we show the $m=4$ case, by choosing $\ket{\Psi_{m}}=\ket{+}_1\ket{0}_2\ket{+}_3\ket{0}_4\ket{C_1}_{\{5,6,\cdots,N\}}$, and it can be generalized to any $m$ directly. Similar as the discussion of $W_{b}^{\mathcal{P}_N}$ in Sec.~\ref{Sec:Co:one}, one can find that $\langle P_1 \rangle=1$, $\langle P_2\rangle=2^{-\lfloor\frac{m}{2}\rfloor}$, and $\langle P_1+P_2\rangle=1+2^{-\lfloor\frac{m}{2}\rfloor}$.
\section*{SUPPLEMENTARY Methods 4: Proof of Corollary \ref{Th:2Dmsep}}\label{Sec:Co:four}
Before proving Corollary \ref{Th:2Dmsep}, we introduce the following Lemma.
\begin{lemma}\label{Lem:2Dnum}
Given an $m$-partition $\mathcal{P}_m$ of an $N$-qubit system, with subsystems $\{A_i\}_{i=1}^m$ and $N\geq \frac{m(m-1)}{2}$, there always exists a bipartition of $\mathcal{P}_m$, denoted as $\{A,\bar{A}\}$, such that the qubit number in $A$ satisfies,
\begin{equation}\label{Eq:Lem2D}
m-1\leq |A| \leq N-(m-1).
\end{equation}
\end{lemma}
\begin{proof}
Without loss of generality, we assume that the qubit number of each subsystem is in the increasing order $|A_1|\leq |A_2|\leq \cdots \leq |A_m|$.
Since $|A_i|\geq 1$, one has $|A_m|=N-\sum_{i=1}^{m-1}\leq N-(m-1)$. Suppose $|A_m|\geq m-1$, we can directly select $A_m$ as $A$. Otherwise we choose $A=A_{m-1}\bigcup A_{m}$, and show that $|A|=|A_{m-1}|+|A_{m}|$ satisfies Eq.~\eqref{Eq:Lem2D} by contradiction.

First, suppose $|A|< m-1$, in the case of $m$ being even, one has

\begin{equation}\label{Eq.Lem2De}
\begin{aligned}
N=\sum_{i=1}^{m} |A_i|
\leq \frac{m}{2} (|A_{m-1}|+|A_{m}|)
<\frac{m(m-1)}{2},
\end{aligned}
\end{equation}
where the first inequality is because there are $\frac{m}{2}$ pair of subsystems $A_{2k-1}, A_{2k}$
for $1\leq k \leq \frac{m}{2}$, with $|A_{2k-1}|+|A_{2k}|\leq |A_{m-1}|+|A_{m}|$. It is clear that Eq.~\eqref{Eq.Lem2De} is contradict with the precondition $N\geq \frac{m(m-1)}{2}$. In the odd $m$ case, one can obtain the same contradiction as Eq.~\eqref{Eq.Lem2De},
\begin{equation}\label{Eq.Lem2Dodd}
\begin{aligned}
N=\sum_{i=1}^{m} |A_i|
\leq (m-2)|A_{m-1}| +(|A_{m-1}|+|A_{m}|)
<(m-2)\frac{m-1}{2} +(m-1)=\frac{m(m-1)}{2},
\end{aligned}
\end{equation}
where in the first inequality we apply $|A_i|\leq |A_{m-1}|$ for the first $m-2$ subsystems, the second inequality is due to the fact that $|A_{m-1}|<\frac{m-1}{2}$.

Moreover, suppose $|A|> N-(m-1)$, one has $|A|= N-(m-2)$, since there is at least one qubit in each of the first $m-2$ subsystems. On account of $|A_m|<m-1$, we have
\begin{equation}\label{}
\begin{aligned}
|A_{m-1}|>N-(m-2)-(m-1)=N-1,
\end{aligned}
\end{equation}
which means that $|A_{m-1}|=N$ and it leads to a clear contradiction.
\end{proof}
Then we give the proof of Corollary \ref{Th:2Dmsep} as follows.
\begin{proof}
According to Theorem 2, similar as the 1-D case, one needs to find the solution of the optimization
\begin{equation}\label{}
f_2(m)\equiv \min_{\mathcal{P}_m}\max_{\{A,\bar{A}\}} S(\rho_A)
\end{equation}
where the maximization takes over all possible partitions $\mathcal{P}_m$ with $m$ subsystems, the minimization takes over all bipartitions of $\mathcal{P}_m$, and  $\rho_A=\mathrm{Tr}_{\bar{A}}(\ket{C_2}\bra{C_2})$.
Comparing to the 1-D case, the partition of the 2-D lattice is richer and more complex, thus here instead of finding the exact solution of $f_2(m)$, we give a reliable lower bound of it, i.e.,
\begin{equation}\label{}
f_2(m)\geq \gamma(m) \equiv \left\lceil \frac{-1+\sqrt{1+8(m-1)}}{2}\right\rceil.
\end{equation}
Note that a lower bound of the optimization can also give us a reasonable witness. We also show that this bound is tight, i.e., it exactly equals to the solution of the optimization, as $m\leq 5$.

In the following, we first prove that $f_2(m)\geq \gamma(m)$ and then give explicit partitions to saturate the bound as $m\leq 5$. On account of Lemma \ref{Lem:2Dnum}, for any m-partition $\mathcal{P}_m$, one can always find a bipartition $\{A,\bar{A}\}$ of $\mathcal{P}_m$ such that the qubit number in $A$ satisfying $m-1\leq |A| \leq N-(m-1)$. Without loss of generality, one can consider $m-1\leq |A| \leq \frac{N}{2}$, otherwise we take $\bar{A}$ as $A$.

Generally speaking, the larger the total qubit number $|A|$ contains, the larger the entanglement entropy $S(\rho_A)$ is. For 2-D cluster state, the entanglement entropy satisfies the area law \cite{Eisert2010area}. Given that the qubit number of $A$ satisfies $m-1\leq |A| \leq \frac{N}{2}$, the best way to minimize the entanglement entropy is to reduce the boundary length of it,that is, gather at the corner of the square lattice, as shown in Fig.~\ref{Fig:2Dtriangle}. Suppose that the subsystem $A$ happens to be a  right-angled isosceles triangle with the length of the hypotenuse being $d$. Then the boundary length of $A$, $|\partial A|=d$, and it is not hard to find that $S(\rho_A)=d$. Consider $|A|=m-1$, and we have the relation,
\begin{equation}\label{Eq:2Dtra}
\frac{d(d+1)}{2}=m-1.
\end{equation}
By solving Eq.~\eqref{Eq:2Dtra}, one obtains $d=\frac{-1+\sqrt{1+8(m-1)}}{2}$. For general $m$, the shape of $A$ is not necessarily a triangle, and we should round up the value and get,
\begin{equation}\label{}
S(\rho_A)\geq \left\lceil \frac{-1+\sqrt{1+8(m-1)}}{2}\right\rceil.
\end{equation}

\begin{figure}[thb]
\centering
\resizebox{9cm}{!}{\includegraphics[scale=1]{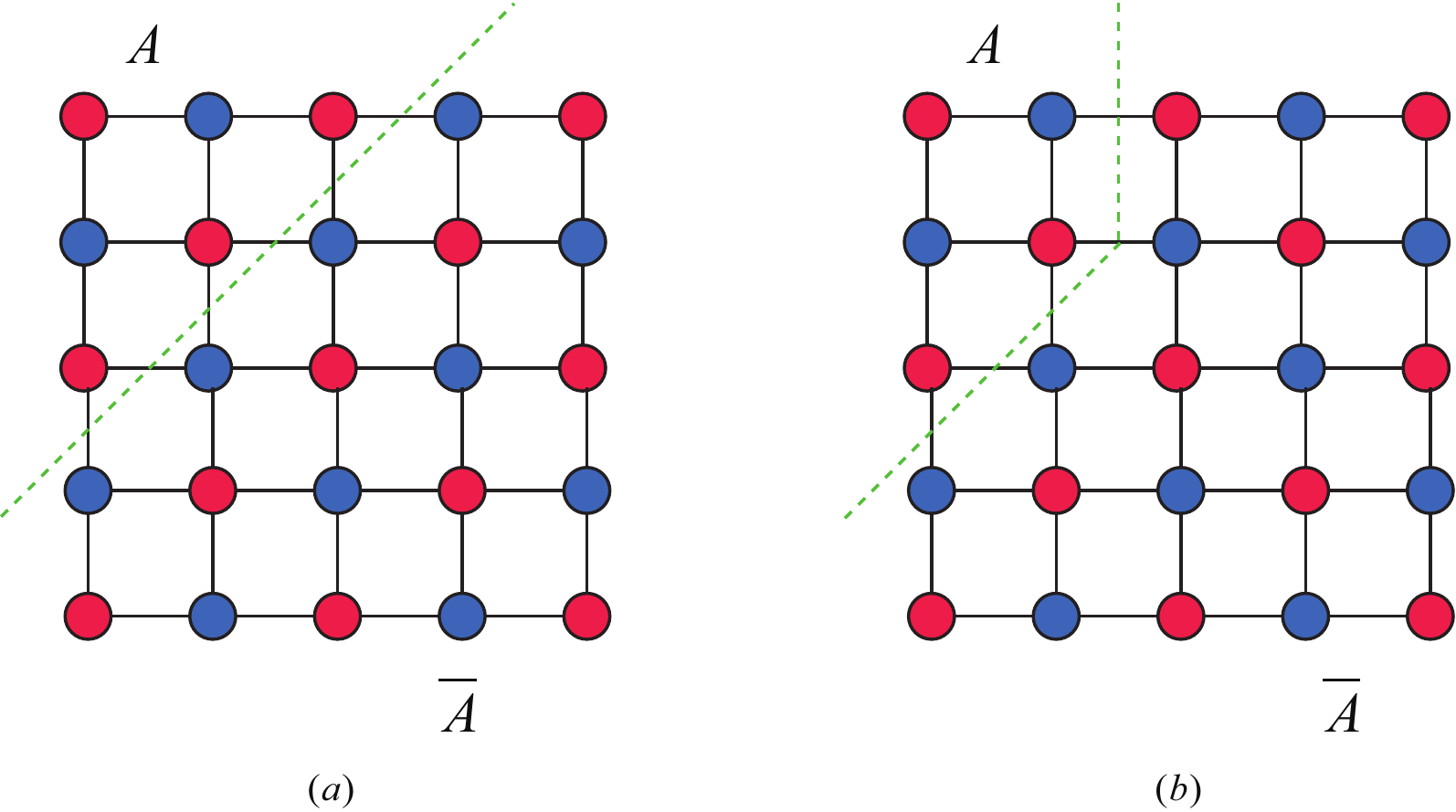}}
\caption{Minimization of the entanglement entropy $S(\rho_A)$ by reducing the boundary length $|\partial A|$. The best strategy is to gather the qubits at the corner of the 2-D lattice. (a)The subsystem $A$ contains 6 qubits and it happens to be a triangle with $S(\rho_A)=|\partial A|=3$. (b)The subsystem $A$ contains 5 qubits and in this case its entropy $S(\rho_A)$ is still $3$.}
\label{Fig:2Dtriangle}
\end{figure}
Consequently, for any $m$-partition $\mathcal{P}_m$, there exists a bipartition $\{A,\bar{A}\}$ of it, such that $S(\rho_A)\geq \gamma(m)$ and hence $\max_{\{A,\bar{A}\}} S(\rho_A)\geq  \gamma(m)$. As a result, one has $f_2(m)\geq \gamma(m)$.

Moreover, we give partitions to saturate this bound as $m\leq 5$. Take the $m=5$ case for example, one can choose the first four subsystems all contain one qubits, i.e., $|A_1|=|A_2|=|A_3|=|A_4|=1$, in a corner of the square lattice, and $A_5$ contains the remaining qubits, as shown in Fig.~\ref{Fig:2Dmsep}. It is not hard to see that $\max_{\{A,\bar{A}\}}S(\rho_A)= \gamma(5)=3$ for any bipartition of this $\mathcal{P}_{m=5}$.
\end{proof}
\begin{figure}[thb]
\centering
\resizebox{3cm}{!}{\includegraphics[scale=1]{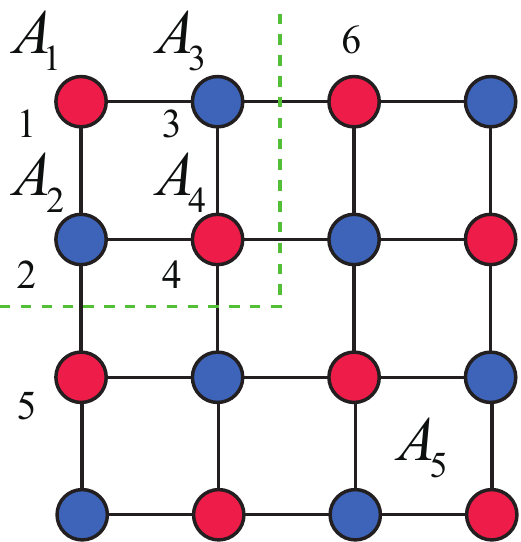}}
\caption{Illustration of a partition that saturates the bound on entanglement entropy, i.e., $S(\rho_A)\geq \gamma(m)=3$, in the $m=5$ case.  The first four subsystems all contain one qubits, i.e., $|A_1|=|A_2|=|A_3|=|A_4|=1$, in the corner of the square lattice, and $A_5$ contains the remaining qubits. It is not hard to see that $\max_{\{A,\bar{A}\}}S(\rho_A)=3$ for any bipartition of this $\mathcal{P}_{m=5}$.}
\label{Fig:2Dmsep}
\end{figure}

Finally, we remark that the non-$m$-separability witnesses shown in Corollary \ref{Th:2Dmsep} is tight or optimal as $m\leq 5$. Here we show a specific $m$-separable state to saturate the witness as $m=5$. The qubit label is given in Fig.~\ref{Fig:2Dmsep}, and we choose $\ket{\Psi_{m}}=\ket{G_{\{1,2,3\}}}\otimes\ket{0}_4\ket{0}_5\ket{0}_6\otimes \ket{G_{\{7,8,\cdots,N\}}}$. Here $\ket{G_{\{1,2,3\}}}\otimes \ket{G_{\{7,8,\cdots,N\}}}$ are the graph state whose graph is obtained from the 2-D lattice by deleting the vertexes $\{4,5,6\}$ and their associated edges.
Similar as the discussion of $W_{b}^{\mathcal{P}_N}$ in Sec.~\ref{Sec:Co:one}, one can find that $\langle P_1 \rangle=2^{-3}$, $\langle P_2\rangle=1$, and $\langle P_1+P_2\rangle=1+2^{-3}$.

\bibliographystyle{apsrev4-1}
%%%%%%%%%%%%%%%%%%%%%%%%%%%%%%%%%%%%%%%%

%%%%%%%%%%%%%%%%%%%%%%%%%%%%%%%%%%%%%%%%
% choose a .bib file
\bibliography{BibSubsystemEW}
\end{document}